\documentclass[12pt,reqno,a4paper]{amsart}
\usepackage{amsmath,amssymb,dsfont,graphicx,cite,latexsym,
epsf,cancel,epic,eepic}
\usepackage[colorlinks=false]{hyperref}
\usepackage{mathpazo}
\usepackage{tikz-cd}

\newtheorem{theorem}{Theorem}
\newtheorem{proposition}{Proposition}
\newtheorem{rmk}{Remark}
\newtheorem{definition}{Definition}

\newcommand{\bbC}{{\mathbb C}}
\newcommand{\bbP}{{\mathbb P}}

\newcommand{\cC}{{\mathcal C}}
\newcommand{\cE}{{\mathcal E}}
\newcommand{\cI}{{\mathcal I}}
\newcommand{\cK}{{\mathcal K}}

\newcommand{\cP}{{\mathcal P}}
\newcommand{\cQ}{{\mathcal Q}}

\newcommand{\X}{{\bf x}}

\def\t{\widetilde}

\title[{Manin involutions for elliptic pencils and discrete integrable systems}]{Manin involutions for elliptic pencils\\ and discrete integrable systems}
\author{Matteo Petrera, Yuri B. Suris, Kangning Wei \and Ren\'e Zander }

\thanks{E-mail: {\tt  petrera.matteo@googlemail.com, suris@math.tu-berlin.de, weikangning12@gmail.com, zander@math.tu-berlin.de}}

\begin{document}

\maketitle

\begin{center}
{\footnotesize{
Institut f\"ur Mathematik, MA 7-1\\
Technische Universit\"at Berlin, Str. des 17. Juni 136,
10623 Berlin, Germany
}}
\end{center}

\begin{abstract}
We contribute to the algebraic-geometric study of discrete integrable systems generated by planar birational maps: 

(a) we
find geometric description of Manin involutions for elliptic pencils consisting of curves of higher degree, birationally equivalent to cubic pencils (Halphen pencils of index 1), and 

(b) we characterize special geometry of base points ensuring that certain compositions of Manin involutions are integrable maps of low degree (quadratic Cremona maps). In particular, we identify some integrable Kahan discretizations as compositions of Manin involutions for elliptic pencils of higher degree.
\end{abstract}

\section{Introduction}
Intimate relation of the theory of integrable systems to algebraic geometry is well appreciated in these days. In the present paper, we address this relation for a very basic class of integrable systems, namely for discrete integrable systems generated by birational maps of $\mathbb{CP}^2$ with a rational integral of motion and an invariant measure with a rational density (whereas the emphasis is put on the integral of motion). In such a system, orbits are confined to invariant curves (level sets of the integral), and on each invariant curve the map induces an automorphism. 

For general reasons, invariant curves must have genus zero or one, since only in these cases the induced automorphisms on the invariant curves can be of infinite order (non-periodic). Our main object of interest will be rational elliptic surfaces (i.e., surfaces birationally equivalent to a plane, admitting a fibration by elliptic curves). A classification of pencils of elliptic curves in a plane was given by Bertini, a modern proof of this result is due to Dolgachev \cite{Dolgachev}. It says that any such pencil is birationally equivalent to a Halphen pencil of index $m\in\mathbb N$, in which a generic curve is of degree $3m$ and has multiplicity $m$ at each of nine base points.

Planar maps preserving pencils of elliptic curves appeared over and over again in the theory of discrete integrable systems. Probably, the most prominent example is given by QRT maps \cite{QRT, Tsuda, Dui}, which preserve pencils of biquadratic curves. Further examples are given by (autonomous versions of) discrete Painlev\'e equations \cite{Sakai,  KNY}, as well as the so called HKY maps which preserve pencils of curves of higher degrees \cite{HKY}. Recently, further examples appeared in the context of the so called Kahan discretization \cite{K}, \cite{PPS2, PZ}, \cite{CMOQ1, CMOQ2, CMOQ4}. A sort of a classification of such maps, based on the Dolgachev's classification of rational elliptic surfaces, was given in \cite{CT1} and sounds almost tautologically: a birational map preserving an Halphen pencil (of index $m$) either preserves each fiber or interchanges the fibers in a nontrivial way.

In the present paper, we are occupied with a construction of integrable maps preserving a pencil of elliptic curves, based only on the pencil itself. The basic idea is to compose two (non-commuting) {\em birational involutions} preserving the pencil. This construction is almost obvious for QRT maps, where one can always use the horizontal and the vertical switches as suitable involutions (the horizontal switch assigns to any point of a biquadratic curve the second intersection point of the curve with the horizontal line through the original point; the definition of the vertical switch is analogous). In \cite{KCMMOQ, KMQ} this construction was extended to a (projectively equivalent) case of maps preserving a pencil of quartic curves with two double points, and it was observed that the corresponding involutions are closely related to the so called {\em Manin involutions}, which were introduced in \cite{Manin} for pencils of cubic curves. Recall that, for a pencil of cubic curves with nine base points $p_1,\ldots,p_9$, the Manin involution centered at $p_1$ assigns to each point $p$ the third intersection point of the line $(p_1p)$ with the curve of the pencil passing through $p$. In \cite{PSS}, it was shown that Kahan discretizations of canonical Hamiltonian systems with cubic Hamiltonian can be characterized as compositions of Manin involutions in the case of a special geometry of the base points of the cubic pencil, see \cite{PS4} for a related result. We also mention that compositions of Manin involutions for cubic pencils appeared in the theory of discrete integrable systems already in \cite{VGR, KMNOY}.

Birational involutions of the plane are classical and well studied objects in algebraic geometry. Their classification was given by Bertini \cite{Bertini} and says that every non-trivial birational involution of $\mathbb P^2$ is birationally conjugate to exactly one of the following: a de Jonqui\`eres involution of degree $d>2$ (which fixes an irreducible curve of degree $d$ with a unique singular point $p$ which is an ordinary multiple point of multiplicity $d-2$), a Geiser involution, or a Bertini involution.  We refrain from giving precise definitions here and refer the reader to the original source. A modern proof of the Bertini classification was given in \cite{BB}. The place of Manin involutions in the Bertini classification is as follows: a Manin involution centered at $p_1$ for a pencil of cubic curves with nine base points $p_1,\ldots,p_9$ is a de Jonqui\`eres involution with center $p_1$ defined by a hyperelliptic curve of degree $d=5$ and of genus 3, with a triple point $p_1$ and eight Weierstrass points $p_2,\ldots,p_9$ (we are indebted to I. Dolgachev for explaining this).

In the present paper, our goal is to use Manin involutions for elliptic pencils to construct integrable dynamical systems. We start by finding a geometric formulation of Manin involutions for elliptic pencils consisting of curves of higher degree. This is achieved via a birational conjugation from the canonical construction for cubic pencils. It seems, however, that the geometric construction directly in terms of the original pencils is not available in the literature. The most non-trivial contribution consists in  finding the conditions under which the involutions and their compositions are of low degree, thus producing simple and attractive examples of integrable birational maps.

\section{Elliptic pencils and cubic pencils}

Throughout the paper, we work over the field $\mathbb C$. We consider pencils of curves in $\bbP^2$, i.e., families of curves $\cP=\{C_\lambda\}$ parametrized by $\lambda\in\bbP^1$,
$$
C_\lambda=\Big\{[x_0:x_1:x_2]\in\mathbb{P}^2: F(x_0,x_1,x_2)+\lambda G(x_0,x_1,x_2)=0\Big\}.
$$
Here $F$, $G$ are two homogeneous polynomials of degree $d$.  The points of the set
$$
B=\Big\{[x_0:x_1:x_2]\in\mathbb{P}^2: F(x_0,x_1,x_2)=G(x_0,x_1,x_2)=0\Big\}
$$
are called \emph{base points} of the pencil $\cP$. As usual, they are counted with multiplicities. We will assume that the multiplicities of each base point on both curves $F=0$ and $G=0$ (and then on all curves of the pencil) are the same. The \emph{type} of the pencil is then 
$$
(d;(n_1)^1(n_2)^2(n_3)^3\ldots)
$$
where $d$ is the degree of the curves of the pencil, $n_1$ the number of simple base points, $n_2$ the number of double base points, $n_3$ the number of triple base points and so on. The pencil itself will be denoted by 
$$
\cP\big(d; p_1^{m_1},p_2^{m_2},\ldots p_N^{m_N}\big),
$$
which refers to the degree $d$ and the list of base points $p_i$ with their respective multiplicities $m_i$, so that $N=n_1+n_2+n_3+\ldots$. Multiplicities $m_i=1$ are usually omitted. 

Counting the intersection numbers, we get:
\begin{equation}\label{pencil intersect number}
d^2=\sum_k n_k k^2.
\end{equation}
Through any point $[x_0:x_1:x_2]\in \bbP^2\setminus B$, there passes a unique curve $C_\lambda$ of the pencil, with $\lambda=-F(x_0,x_1,x_2)/G(x_0,x_1,x_2)$.

Our main interest is in the \emph{elliptic pencils}, for which generic curves of the pencil are of genus $g=1$. According to the degree-genus formula, the genus of irreducible curves of the pencil is given by:
\begin{equation}\label{pencil genus}
g=\frac{(d-1)(d-2)}{2}-\sum_k n_k\frac{k(k-1)}{2}=1.
\end{equation}
We remark that by virtue of \eqref{pencil intersect number},  the latter equation is equivalent to
\begin{equation}\label{from pencil genus}
3d=\sum_k n_kk,
\end{equation}
where the right-hand side is the total number of base points (counted with multiplicities).

Examples:
\begin{enumerate}
\item \emph{A pencil of the type $(3;9^1)$} of cubic curves with nine simple base points.
\item \emph{A pencil of the type $(4;8^12^2)$} of curves of degree 4 with eight simple and two double points. By an automorphism of $\bbP^2$, we can send the double points to infinity (say, to $[0:1:0]$ and $[0:0:1]$), then in the affine coordinates $(x_1/x_0, x_2/x_0)$, we get a pencil of biquadratic curves. Such pencils are pretty well studied and have plenty of applications in the theory of discrete integrable systems \cite{QRT, Dui}.
\item \emph{A pencil of the type $(6;6^13^22^3)$} of curves of degree 6 with six simple points, three double points and two triple points.
\item \emph{A pencil of the type $(6;9^2)$} of curves of degree 6 with nine double points, i.e., a Halphen pencil of index 2. 
\end{enumerate}

{\bf Remark.} We do allow \emph{infinitely near base points}, at which the curves of the pencil have to satisfy certain conditions tangency up to certain order. In the formulations of our general results about geometry of Manin involutions, we silently assume that the geometry of base points is generic, in particular, that there are no incidental collinearities. However, all our main examples are non-generic with a plenty of incidental collinearities, since it is exactly this feature that allows for a substantial drop of degree of the resulting birational maps. We hope this will not lead to any confusions.


\section{Manin involutions}
For cubic curves, one has a simple geometric interpretation of the addition law. Correspondingly, there is a simple geometric construction of certain birational involutions of $\bbP^2$ induced by pencils of cubic curves, cf. \cite[p. 1376]{Manin}, \cite[p. 35]{Ves}. These were dubbed \emph{Manin involutions} in \cite[Sect. 4.2]{Dui}.

\begin{definition} \label{def Manin cubic} \textbf{(Manin involutions for cubic pencils)}

{\em 1)}
Consider a nonsingular cubic curve $C$ in $\bbP^2$, and a point $p_0\in C$. The {\em Manin involution} on $C$ with respect to $p_0$ is the map $I_{C, p_0}: C\rightarrow C$ defined as follows: for a generic $p\neq p_0$, the image $I_{C, p_0}(p)$ is the unique third intersection point of $C$ with the line $(p_0p)$; for $p=p_0$, the line $(p_0p)$ should be interpreted as the tangent line to $C$ at $p_0$.

{\em 2)}   Consider a pencil $\cP=\{C_\lambda\}$ of cubic curves in $\bbP^2$ with at least one nonsingular member. Let $p_0$ be a base point of the pencil. The {\em Manin involution} $I_{\cP,p_0}:\bbP^2\dashrightarrow\bbP^2$ is a birational map defined as follows. For any $p\in\bbP^2$ which is not a base point, $I_{\cP,p_0}(p)=I_{C_\lambda,p_0}(p)$, where $C_\lambda$ is the unique curve of the pencil through the point $p$. 
\end{definition}

For elliptic pencils of degree higher than 3, the geometric construction of Manin involutions seems to be unknown. The only exception are the vertical and the horizontal switches in biquadratic pencils, of which the famous QRT maps are composed \cite{QRT, Dui}. They can be immediately translated to a construction of \emph{generalized Manin involutions} for quartic pencils with two double points, with respect to the both double points \cite{KMQ}. The definition of the generalized Manin involution $I_{C, p_0}$ for a quartic curve $C$ and a double point $p_0\in C$, resp. of the generalized Manin involution $I_{\cP,p_0}$ for a quartic pencil with two double points, one of them being $p_0$, literally coincides with Definition \ref{def Manin cubic}. This is justified by the fact that any line through a double point $p_0\in C$ still intersects the quartic curve $C$ at two further points.

The main goal of this paper is to elaborate on the geometric definition of Manin involutions in arbitrary elliptic pencils birationally equivalent to a Halphen pencil of index 1, i.e., to a cubic pencil. 

To find a birational equivalence, one can resolve the multiple base points  by means of suitable birational transformations. Often, the simplest way of doing this is by a sequence of \emph{quadratic Cremona transformations}. Recall that a generic quadratic Cremona transformation $\phi:\bbP_1^2\dashrightarrow\bbP_2^2$ has three distinct fundamental points $\mathcal I(\phi)=\{p_1,p_2,p_3\}$ which are blown up to three lines $(q_2q_3)$, $(q_3q_1)$, $(q_1q_2)$, respectively. The three lines $(p_2p_3)$, $(p_3p_1)$, $(p_1p_2)$ are blown down to the points $q_1$, $q_2$, $q_3$, respectively, which build the indeterminacy set of the inverse map: $\mathcal I(\phi^{-1})=\{q_1,q_2,q_3\}$. A practical way to construct such a map consists in finding homogeneous polynomials $\phi(x_0,x_1,x_2)$ of degree 2 vanishing at the fundamental points $p_1,p_2,p_3$. Geometrically, we are speaking about the set of conics in $\bbP_1^2$ through $p_1,p_2,p_3$. The space of solutions of this linear system is two-dimensional: $\alpha \phi_0+\beta \phi_1+\gamma \phi_2$, where $\phi_0,\phi_1,\phi_2$ are homogeneous polynomials of $(x_0,x_1,x_2)$ of degree 2. The map  
\begin{equation}\label{birational change coord}
\phi:[x_0:x_1:x_2]\mapsto [u_0:u_1:u_2]=[\phi_0(x_0,x_1,x_2):\phi_1(x_0,x_1,x_2):\phi_2(x_0,x_1,x_2)]
\end{equation}
is the sought after birational map $\bbP^2\dashrightarrow\bbP^2$. A different choice of a basis $\phi_0,\phi_1,\phi_2$ of the net corresponds to a linear projective transformation of the target plane $\bbP_2^2$. 

Note that the pre-image of a generic line $au_0+bu_1+cu_2=0$ in the target plane $\bbP_2^2$ is the conic $a\phi_0+b\phi_1+c\phi_2=0$ (passing through $p_1,p_2,p_3$) in the source plane $\bbP_1^2$. It follows that for any regular point $p$ of $\phi$, the pencil of lines $\cP(1;q)$ through $q=\phi(p)$ in $\bbP_2^2$ corresponds to the pencil of conics $\cP(2; p,p_1,p_2,p_3)$ in $\bbP_1^2$.

\section{Example: a quartic pencil with two double base points}
\label{Sect quartic}

\subsection{Geometry of the base points}

Consider an elliptic pencil in $\bbP^2$ of the type $(4;8^12^2)$,
$$
\cE=\cP(4;p_1,\ldots,p_8,p_9^2,p_{10}^2).
$$
Thus, $\cE$ consists of quartic curves with 8 simple base points $p_1,\ldots,p_8$ and two double base points $p_9,p_{10}$.
The position of the ten base points is not arbitrary: for a generic configuration of ten points, there exists just one curve of degree 4 through these points, having the prescribed two of them as double points. On the other hand, for a generic configuration of nine points, there is a one-parameter family (a pencil) of curves of degree 4 through these points, having the prescribed two of them as double points (nine incidence conditions plus four second order conditions, altogether 13 linear conditions, while a generic curve of degree 4 has 14 non-homogeneous coefficients). Counting the intersection numbers, we see that all curves of the pencil pass through a further simple point (indeed, seven simple points and two double points contribute $7\times 1+2\times 4=15$ to the intersection number $16$).

More information on the configuration of the ten base points is contained in the following statement.
\begin{proposition}
In a generic pencil $\cP(4;p_1,\ldots,p_8,p_9^2,p_{10}^2)$, one of the curves is reducible and consists of the line $(p_9p_{10})$ and a cubic curve passing through all ten base points $p_1,\ldots,p_{10}$.
\end{proposition}
\begin{proof}
Fix any point $p\in(p_9p_{10})$  different from $p_9$, $p_{10}$, and consider the unique curve $C$ of the pencil through $p$. If  the line $(p_9p_{10})$ would not be a component of this curve, then the intersection number of $C$  with the line $(p_9p_{10})$ would be at least $2\times 2+1=5$, a contradiction. Thus, the curve $C$ is reducible and contains the line $(p_9p_{10})$ as one of the components. Another component is a cubic curve through $p_1,\ldots,p_{10}$ (with $p_9$, $p_{10}$ being simple points on the cubic).
\end{proof}

\begin{rmk}
If the reducible curve $C$ happens to contain $(p_9p_{10})$ as a double line, then the remaining component is a conic through eight base points $p_1,\ldots,p_8$. 
\end{rmk}

\subsection{Birational reduction to a cubic pencil}

Consider a pencil $\cE=\cP(4;p_1,\ldots,p_8,p_9^2,p_{10}^2)$. Let  $\phi:\bbP_1^2\dashrightarrow\bbP_2^2$ be a quadratic Cremona map with the fundamental points $p_1,p_9,p_{10}$.  Thus, $\phi$ blows down the lines  $(p_9p_{10})$, $(p_1p_{10})$, $(p_1p_9)$ to points denoted by $q_1,q_9,q_{10}$, respectively, and blows up the points $p_1,p_{9},p_{10}$ to the lines $(q_{9}q_{10})$, $(q_1q_{10})$, $(q_1q_{9})$. All other base points $p_i$, $i=2,\ldots,8$, are regular points of $\phi$, their images will be denoted by $q_i=\phi(p_i)$. 

\begin{proposition} \label{th quartic to cubic} 
Under the map $\phi$:
\begin{itemize}
    \item[a)] Quartic curves of the original pencil $\cE$ in $\bbP_1^2$ correspond to curves of a cubic pencil 
    $$
    \cP(3;q_2, \ldots, q_8, q_9, q_{10})
    $$
    with nine base points in $\bbP_2^2$; the point $q_1$ is not a base point of the latter pencil. 
    \item[b)] For $i=2,\ldots,8$, the pencil of lines $\cP(1;q_i)$ in $\bbP_2^2$ corresponds to the pencil of conics 
    $$
    \cP(2; p_i,p_1, p_9, p_{10})
    $$
    in $\bbP_1^2$.    
    \item[c)] The pencils of lines $\cP(1;q_9)$, $\cP(1; q_{10})$ in $\bbP_2^2$ correspond to the pencils of lines 
    $$
    \cP(1;p_9), \quad \cP(1; p_{10})
    $$
    in $\bbP_1^2$. 
    \end{itemize}
\end{proposition}
\begin{proof} 
\quad
\begin{itemize}
   \item[a)]  The total image of a quartic curve $C\in\cE$ is a curve of degree 8. Since $C$ passes through $p_1$, its total image contains the line $(q_9q_{10})$. Since $C$ passes through $p_9$ and $p_{10}$ with multiplicity 2, its total image contains the lines $(q_1q_{10})$ and $(q_1q_9)$ with multiplicity 2. Dividing by the linear defining polynomials of all these lines, we see that the proper image of $C$ is a curve of degree $8-5=3$. This curve has to pass through all points $q_i$, $i=2,\ldots,8$. The curve $C$ of degree 4 has no other intersections with the line $(p_9p_{10})$ different from two double points $p_9$ and $p_{10}$, therefore its proper image does not pass through $q_1$. On the other hand, the curve $C$ of degree 4 has one additional intersection point with each of the lines $(p_1p_9)$ and $(p_1p_{10})$, different from the simple point $p_1$ and the double point $p_9$, resp. $p_{10}$. Therefore, its proper image passes through $q_{10}$, resp. $q_9$, with multiplicity 1. 
    
    \item[b)] This follows from the fact that $p_i$, $i=2,\ldots,8$, are regular points of $\phi$.
     
    \item[c)] Consider the total pre-image of a line through $q_9$. It is a conic through $p_1, p_9, p_{10}$ whose defining polynomial vanishes on the line $(p_1p_{10})$.  Thus, the conic is reducible and contains that line. Dividing by the defining polynomial of this line (of degree 1), we see that the proper pre-image is a line which must pass through $p_9$. Similarly, the proper pre-image of a line through $q_{10}$ is a line through $p_{10}$. 
           
\end{itemize}
\end{proof}

Let $S$ be the elliptic surface obtained from $\bbP^2$ by blowing up the ten base points $p_i$, $i=1,\ldots,10$. Let $E_i$ be the exceptional divisor classes of the blow-ups. The Picard group of $S$ is ${\rm Pic}(S)=\mathbb{Z}D\bigoplus\mathbb{Z}E_1\bigoplus\ldots\bigoplus\mathbb{Z}E_{10}$. The class of a generic curve of the pencil $\cE$ is 
\begin{equation} \label{biquadr divisor}
4D-E_1-E_2-E_3-E_4-E_5-E_6-E_7-E_8-2E_9-2E_{10}.
\end{equation}
The quadratic Cremona map of Proposition \ref{th quartic to cubic} corresponds to the following change of basis of the Picard group:
\begin{equation} \label{biquadr change Picard}
\renewcommand{\arraystretch}{1.3}
\left\{ \begin{array}{ccl}
D' & = & 2D-E_1-E_9-E_{10},\\ 
E'_1 & = & D-E_9-E_{10}, \\ 
E'_9 & = & D-E_1-E_{10},\\ 
E'_{10} & = & D-E_1-E_9,
\end{array} \right.
\end{equation}
and $E'_i=E_i$ for $i=2,\ldots,8$.

One can check that $E'_1$ is a redundant class, in the sense that the class \eqref{biquadr divisor} of a general curve of the pencil is expressed through $E'_2,\ldots,E'_{10}$ only:
\begin{equation} \label{K}
4D-E_1-\ldots-E_8-2E_9-2E_{10}=3D'-E'_2-E'_3-\ldots-E'_{10}.
\end{equation} 
This corresponds to the fact that $q_1$ is not a base point of the $\phi$-image of the pencil $\cE$. Note that $E'_1=D-E_9-E_{10}$ is the class of (the proper transform of) the line $(p_9p_{10})$ in $\bbP^2$.  Blowing down $E'_1$ on $S$, we obtain the surface $S^\prime$ which is a minimal elliptic surface (blow-up of $\mathbb P^2$ at nine points), whose anti-canonical divisor class coincides with \eqref{K}. 

Statement b) of Proposition \ref{th quartic to cubic} translates to relations $D'-E'_i=2D-E_1-E_9-E_{10}-E_i$ in the Picard group (for $i=2,\ldots,8$), while statement c) translates as  $D'-E'_9=D-E_9$ and $D'-E'_{10}=D-E_{10}$.

\subsection{Manin involutions}

In the new coordinates, where the pencil consists of cubic curves, Manin involutions $I_{q_i}$ with respect to the base points $q_i$ of the pencil are defined as in Definition \ref{def Manin cubic}: for a point $q$ which is not a base point, $I_{q_i}(q)$ is the unique third intersection of the the line $(q_iq)$ with the cubic curve of the pencil passing through $q$. We now pull back this construction to the original pencil in the old coordinates.

\begin{definition} \textbf{(Manin involutions for pencils of the type $(4;8^1 2^2)$)}
\smallskip

Consider a pencil $\cE=\cP(4;p_1,\ldots,p_8,p_9^2,p_{10}^2)$. There are two kinds of Manin involutions.
\smallskip

 {\em 1)} Involutions $I_{i,j}^{(2)}$, $i,j\in\{1,\ldots,8\}$, defined in terms of the pencil  of conics 
 $$
\cC_{i,j}= \cP(2; p_i,p_j,p_9,p_{10}).
 $$
Given a point $p$ which is not a base point of $\cE$, there is a unique conic of $\cC_{i,j}$ passing through $p$ and a unique quartic curve of $\cE$ passing through $p$. We set $I_{i,j}^{(2)}(p)=p'$, where $p'$ is the unique further intersection point of those two curves. This intersection is unique, since the intersection number of the conic with the quartic is $2\times 4=8$, while the intersections at the points $p_i,p_j,p_9,p_{10}$, and $p$ count as $1+1+2+2+1=7$. 
 \smallskip
  
  {\em 2)} Involutions $I_{9}^{(1)}$, $I_{10}^{(1)}$ defined in terms of the pencils of lines
$$
\cP(1; p_9),\;\;resp.\;\; \cP(1; p_{10}).
$$ 
For instance, the involution $I_9$ is defined as follows. Given a point $p$ which is not a base point of $\cE$, we set $I_{9}^{(1)}(p)=p'$, where $p'$ is the unique third intersection of the line $(p_9p)$ and the quartic curve of $\cE$ passing through $p$. This intersection is unique, since $p_9$ is a double point of the curve. 
 \end{definition}

Indeed:   
\begin{itemize}
 
 \item[1)] Due to point b) of Proposition \ref{th quartic to cubic}, for any $i=2,\ldots,8$, the Manin involution with respect to $q_i$ is conjugated to the map defined as above in terms of conics through $p_1$, $p_9$, $p_{10},$ and $p_i$. Remarkably, while in the construction of the conjugating Cremona map the roles of the simple base points $p_1$ and $p_i$ are asymmetric, in the resulting map  $I_{1,i}^{(2)}$  the points $p_1$ and $p_i$ are on an equal footing. More generally, $I_{i,j}^{(2)}=I_{j,i}^{(2)}$, where the map on the left-hand side should be understood as conjugated to $I_{q_j}$ under the quadratic Cremona map with the fundamental points $p_i,p_9,p_{10}$, while the right-hand side should be understood as conjugated to   $I_{q_i}$ under the quadratic Cremona map with the fundamental points $p_j,p_9,p_{10}$.

 \item[2)] Due to point c) of Proposition \ref{th quartic to cubic}, Manin involutions $I_{q_9}, I_{q_{10}}$ on $\bbP_2^2$ are conjugated to the maps $I_{9}^{(1)}, I_{10}^{(1)}$ on $\bbP_1^2$ defined in terms of lines through $p_9, p_{10}$, respectively. Again, while the construction depends on the choice of a simple base point $p_1$, the resulting map does not depend on this choice.
\end{itemize}

The involution $I_{i,j}^{(2)}$ has all base points of the pencil as singularities (indeterminacy points). For instance, it blows up the point  $p_k$ to the conic through $p_i,p_j,p_k,p_9,p_{10}$. However, a composition 
$$
I_{j,k}^{(2)} \circ I_{i,j}^{(2)}
$$ 
with three distinct simple base points $p_i,p_j,p_k$  is well defined at $p_k$ and maps it to $p_i$. Moreover, this composition can be characterized as the unique map acting on the elliptic curves of the pencil as the shift mapping $p_k$ to $p_i$. In particular, this composition does not depend on $j$.


\section{Example: a sextic pencil with three double base points and two triple base points}
\label{Sect sextic}

\subsection{Birational reduction to a cubic pencil} Consider an elliptic pencil in $\bbP^2$ of the type $(6;6^1 3^2 2^3)$,
$$
\cE=\cP(6;p_1,\ldots,p_6,p_7^2,p_8^2,p_9^2,p_{10}^3,p_{11}^3), 
$$
consisting of curves of degree 6 with six simple base points $p_1,\ldots,p_6$, three double base points $p_7,p_8,p_9$, and two triple base points $p_{10},p_{11}$. 
We reduce it to a cubic pencil in two steps.
\medskip 

\emph{Step 1.} Apply a quadratic Cremona map $\phi'$ with the fundamental points $p_9,p_{10},p_{11}$ (the both triple base points and one of the double base points).  Thus, $\phi'$ blows down the lines  $(p_{10}p_{11})$, $(p_9p_{11})$, $(p_9p_{10})$ to the points denoted by $q_9,q_{10},q_{11}$, respectively, and blows up the points $p_9,p_{10},p_{11}$ to the lines $(q_{10}q_{11})$, $(q_9q_{11})$, $(q_9q_{10})$. All other  base points $p_i$, $i=1,\ldots,8$ are regular points of $\phi'$ and their images are denoted by $q_i=\phi'(p_i)$.

\begin{proposition} \label{th sextic to quartic}
The change of variables $\phi'$ maps a pencil $\cE=\cP(6;p_1,\ldots,p_6,p_7^2,p_8^2,p_9^2,p_{10}^3, p_{11}^3)$ of sextic curves to a pencil $\cP(4;q_1,\ldots,q_6,q_{10},q_{11},q_7^2,q_8^2)$ of quartic curves with eight simple base points and two double base points. The point $q_9$ is not a base point of the latter pencil.
\end{proposition}
\begin{proof} 
 The total image of a curve $C\in\cE$ is a curve of degree 12. Since $C$ passes through $p_9,p_{10},p_{11}$ with the multiplicities 2,3,3, its total image contains the lines $(q_{10}q_{11})$, $(q_9q_{11})$, $(q_9q_{10})$ with the same multiplicities. Dividing by the linear defining polynomials of all these lines, we see that the proper image of $C$ is a curve of degree $12-8=4$. This curve passes through all points $q_i$, $i=1,\ldots,8$ (for $i=7,8$ with multiplicity 2). The curve $C$ of degree 6 has no other intersections with the line $(p_{10}p_{11})$ different from two triple points $p_{10}$ and $p_{11}$, therefore its proper image does not pass through $q_9$. On the other hand, the curve $C$ of degree 6 has one additional intersection point with each of the lines $(p_9p_{10})$ and $(p_9p_{11})$, different from the double point $p_9$ and the triple point $p_{10}$, respectively $p_{11}$. Therefore, its proper image passes through $q_{11}$, resp. $q_{10}$, with multiplicity 1. 
\end{proof}
\medskip

\emph{Step 2.} Apply a quadratic Cremona map $\phi''$ with the fundamental points $q_7,q_8$ (the both double base points), and one of the simple base points. As we know from Proposition \ref{th quartic to cubic}, the image of the pencil $\cP(4;q_1,\ldots,q_6,q_{10},q_{11},q_7^2,q_8^2)$ under $\phi''$ is a pencil of cubic curves with nine base points. The nature of the composition $\phi''\circ \phi'$ depends on the choice of the simple base point $q_i$ designated as the third fundamental point of $\phi''$, and is different in the cases $i=1,\ldots,6$ and $i=10,11$. It turns out that the first option contains all the possibilities for the different sorts of Manin involutions, therefore we restrict our attention to this case, taking, for definiteness, $i=6$. Thus, let $\phi''$ have three fundamental points $q_6,q_7,q_8$. It blows down the lines  $(q_6q_7)$, $(q_6q_8)$, $(q_7q_8)$ to points $r_8,r_7,r_6$, respectively, and blows up the points $q_6,q_7,q_8$ to the lines $(r_7r_8)$, $(r_6r_8)$, $(r_6r_7)$. All other base points $q_i$, $i=1,\dotsc,5,10,11$ are regular points of $\phi’’$, their images will be denoted by $r_i=\phi’’(q_i)$. As follows from Propositions \ref{th sextic to quartic}, \ref{th quartic to cubic}, we have:
\begin{proposition} \label{th sextic vs cubic b}
The change of coordinates $\phi=\phi''\circ \phi':\bbP_1^2\dashrightarrow\bbP_2^2$ maps a pencil $$\cE=\cP(6;p_1,\ldots,p_6,p_7^2,p_8^2,p_9^2,p_{10}^3, p_{11}^3)$$ of sextic curves in $\bbP_1^2$ to a pencil $$\cP(3; r_1, \ldots, r_5,r_7,r_8, r_{10},r_{11})$$ of cubic curves with nine base points in $\bbP_2^2$. The points $r_6$ and $r_9$ are not base points of this cubic pencil. 
\end{proposition}

Properties of the  birational change of coordinates $\phi=\phi''\circ \phi'$ on $\bbP^2$ are easily obtained.
It is a Cremona map of degree 4 which blows down the lines $(p_9p_{10})$, $(p_9p_{11})$, $(p_{10}p_{11})$ to the points $r_{11}$, $r_{10}$, $r_9$, respectively, and blows down the conics $C(p_6,p_7,p_9,p_{10},p_{11})$,  $C(p_6,p_8,p_9,p_{10},p_{11})$, $C(p_7,p_8,p_9,p_{10},p_{11})$ to the points $r_8$, $r_7$, $r_6$, respectively. Moreover, $\phi$ blows up the points $p_9$, $p_{10}$, $p_{11}$ to the lines $(r_{10}r_{11})$, $(r_9r_{11})$, $(r_{9}r_{10})$, respectively, and the points $p_6$, $p_7$, $p_8$ to the conics $C(r_7,r_8,r_9,r_{10},r_{11})$,  $C(r_6,r_8,r_9,r_{10},r_{11})$, $C(r_6,r_7,r_9,r_{10},r_{11})$, respectively. Points $p_i$, $i=1,\ldots,5$, are regular points of $\phi$, their images are $r_i=\phi(p_i)$. The pre-image of a generic line in $\bbP^2$ is a quartic curve passing through $p_6,\ldots,p_{11}$ (the points $p_{10}$ and $p_{11}$ being of multiplicity 2). In particular, for any regular point $p$, the pencil of lines $\cP(1;r)$ through $r=\phi(p)$ in $\bbP_2^2$ corresponds to the pencil 
    $$
    \cP(4; p,p_6,p_7,p_8,p_9^2,p_{10}^2,p_{11}^2)
    $$  
    of quartic curves in $\bbP_1^2$ .  
    
\begin{proposition} \label{th sextic vs cubic}
The change of coordinates $\phi=\phi''\circ \phi':\bbP_1^2\dashrightarrow\bbP_2^2$  has the following properties:
\begin{itemize}
    \item[a)] For $i=1,\ldots,5$, the pencil of lines $\cP(1;r_i)$ in $\bbP_2^2$ corresponds to the pencil 
    $$
    \cP(4; p_i,p_6,p_7,p_8,p_9^2,p_{10}^2,p_{11}^2)
    $$  
    of quartic curves in $\bbP_1^2$.
    
    \item[b)] For $i=10,11$, the proper pre-images of lines of the pencil $\cP(1;r_i)$ in $\bbP_2^2$ are cubics of the respective pencil
    $$
    \cP(3; p_6,p_7,p_8,p_9,p_{10}^2,p_{11}), \quad  \cP(3; p_6,p_7,p_8,p_9,p_{10},p_{11}^2)
    $$
    in $\bbP_1^2$.  
      
    \item[c)] For $i=7,8$, the proper pre-images of lines of the pencil $\cP(1;r_i)$ in $\bbP_2^2$ are conics of the respective pencil 
    $$
    \cP(2;p_7,p_9,p_{10},p_{11}), \quad \cP(2;p_8,p_9,p_{10},p_{11})
    $$
    in $\bbP_1^2$. 
\end{itemize}
\end{proposition}

\begin{proof} \quad
\begin{itemize}
   \item[a)] This follows from the fact that $p_i$, $i=1,\ldots,5$ are regular points of $\phi$.
     \medskip
     
  \item[b)] Consider the total pre-image of a line through $r_{10}$. It is a quartic curve passing through $p_6,\ldots,p_{11}$, having $p_{10},p_{11}$ as double points.  Its defining polynomial vanishes on the line $(p_9p_{11})$, which blows down to $r_{10}$.  Thus, the quartic is reducible and contains that line. Dividing by the defining polynomial of the line, we see that the proper pre-image is a cubic passing through $p_6,p_7,p_8,p_{10},p_{11}$, with $p_{10}$ being a double point. 
  \medskip    
     
  \item[c)] Consider the total pre-image of a line through $r_7$. It is a quartic curve passing through $p_6,\ldots,p_{11}$, having $p_{10},p_{11}$ as double points.  Its defining polynomial vanishes on the conic $C(p_6,p_8,p_9,p_{10},p_{11})$, which blows down to $r_7$.  Thus, the quartic is reducible and contains that conic. Dividing by the defining polynomial of the conic, we see that the proper pre-image is a conic passing through $p_7,p_9,p_{10},p_{11}$. 
\end{itemize}
\end{proof}      
       
Let $S$ be the elliptic surface obtained from $\bbP^2$ by blowing up the eleven base points $p_i$, $i=1,\ldots,11$. Let $E_i$ be exceptional divisor classes of the blow ups. The Picard group of $S$ is ${\rm Pic}(S)=\mathbb{Z}D\bigoplus\mathbb{Z}E_1\bigoplus\ldots\bigoplus\mathbb{Z}E_{11}$. The class of a generic curve of the pencil is 
\begin{equation} \label{sextic divisor}
6D-E_1-E_2-E_3-E_4-E_5-E_6-2E_7-2E_8-2E_9-3E_{10}-3E_{11}.
\end{equation}
The quadratic Cremona map $\phi'$ corresponds to the following change of basis of Pic$(S)$:
\begin{equation} \label{sextic change Picard 1}
\renewcommand{\arraystretch}{1.3}
\left\{ \begin{array}{ccl}
D' & = & 2D-E_9-E_{10}-E_{11},\\ 
E'_9 & = & D-E_{10}-E_{11}, \\ 
E'_{10} & = & D-E_9-E_{11},\\ 
E'_{11} & = & D-E_9-E_{10},
\end{array} \right.
\end{equation}
and $E'_i=E_i$ for $i=1,\ldots,8$. The Cremona map $\phi''$ corresponds to the following change of basis of the Picard group:
\begin{equation} \label{sextic change Picard 2}
\renewcommand{\arraystretch}{1.3}
\left\{ \begin{array}{ccl}
D''& = & 2D'-E'_6-E'_7-E'_8,\\ 
E''_6 & = & D'-E'_7-E'_8, \\ 
E''_7 & = & D'-E'_6-E'_8,\\ 
E''_8 & = & D'-E'_6-E'_7,
\end{array} \right.
\end{equation}
and $E''_i=E'_i$ for $i=1,\ldots,5$ and $i=9,10,11$.
Composing \eqref{sextic change Picard 1}, \eqref{sextic change Picard 2}, we easily compute:
\begin{equation} \label{sextic change Picard}
\renewcommand{\arraystretch}{1.3}
\left\{ \begin{array}{ccl}
D''& = & 4D-E_6-E_7-E_8-2E_9-2E_{10}-2E_{11},\\ 
E''_6 & = & 2D-E_7-E_8-E_9-E_{10}-E_{11}, \\ 
E''_7 & = & 2D-E_6-E_8-E_9-E_{10}-E_{11},\\ 
E''_8 & = & 2D-E_6-E_7-E_9-E_{10}-E_{11},\\ 
E''_9 & = & D-E_{10}-E_{11}, \\ 
E''_{10} & = & D-E_9-E_{11},\\ 
E''_{11} & = & D-E_9-E_{10},
\end{array} \right.
\end{equation}
and $E''_i=E_i$ for $i=1,\ldots,5$. One can check that the classes 
$$
E''_{6}=2D-E_7-E_8-E_9-E_{10}-E_{11}, \quad E''_{9}=D-E_{10}-E_{11}
$$ 
are redundant, in the sense that the class \eqref{sextic divisor} of a general curve of the pencil $\cE$ is expressed through $E''_i$, $i\neq 6,9$:
\begin{align} \label{K sextic}
6D-E_1-\ldots-E_6-2E_7\ -  & \ 2E_8-2E_9- 3E_{10}-3E_{11}  \nonumber\\
 &  =3D''-E''_1-\ldots-E''_5-E''_7-E''_8-E''_{10}-E''_{11}.
\end{align} 
This reflects the fact that $r_6$, $r_9$ are not base points of the resulting cubic pencil.
The redundant classes are the classes (of the proper transforms) of the conic $C(p_7,p_8,p_9,p_{10},p_{11})$, resp. of the line $(p_{10}p_{11})$ in $\bbP_1^2$.  The surface $S'$ obtained by blowing down $E''_6$ and $E''_{9}$ on $S$, is a minimal elliptic surface, whose anti-canonical divisor class coincides with \eqref{K sextic}. Generic fibers of $S'$ are exactly the lifts of generic curves of the initial sextic pencil $\cE$.

Note that statements of Proposition \ref{th sextic vs cubic} translate as the following relations  in Pic$(S)$:
\begin{eqnarray*}
D''-E''_i & = & 4D-E_i-E_6-E_7-E_8-2E_9-2E_{10}-2E_{11}, \quad i=1,\ldots,5,\\
D''-E''_{10} & = & 3D-E_6-E_7-E_8-E_9-2E_{10}-E_{11}, \\
D''-E''_7 & = & 2D-E_7-E_9-E_{10}-E_{11}.
\end{eqnarray*}

\subsection{Manin involutions}

We pull back the standard construction of Manin involutions for the cubic pencil in $\bbP_2^2$ by means of the map $\phi$ to the original pencil in $\bbP_1^2$.

\begin{definition} \textbf{(Manin involutions for pencils of the type $(6;6^1 3^2 2^3)$)} \quad

Consider a pencil $\cE=\cP(6;p_1,\ldots,p_6,p_7^2,p_8^2,p_9^2,p_{10}^3,p_{11}^3)$. There are three kinds of Manin involutions.
\smallskip

{\em 1)} Involutions $I_{i,j,k}^{(4)}$, $i,j\in\{1,\ldots,6\}$, $k\in\{7,8,9\}$. E.g., $I_{i,j,9}^{(4)}$ is defined in terms of quartic curves of the pencil
 $$
\cQ_{i,j,9}=\cP(4; p_i,p_j,p_7,p_8,p_9^2,p_{10}^2,p_{11}^2).
 $$
Given a point $p$ which is not a base point of $\cE$, there is a unique quartic curve of $\cQ_{i,j,9}$ through $p$ and a unique sextic curve of $\cE$ through $p$. We set $I_{i,j,9}^{(4)}(p)=p'$, where $p'$ is the unique further intersection point of these two curves. This intersection is unique, since the intersection number of the quartic with the sextic is $4\times 6=24$, while the intersections at the points $p_i,p_j,p_7,p_8,p_9,p_{10},p_{11}$, and $p$ count as $1+1+2+2+4+6+6+1=23$. Involutions $I_{i,j,k}^{(4)}$ with $k=7,8$ are defined similarly.
\smallskip

 {\em 2)} Involutions $I_{i,k}^{(3)}$, $i\in\{1,\ldots,6\}$, $k\in\{10,11\}$. E.g., $I_{i,10}^{(3)}$ is defined in terms of cubic curves of the pencil
 $$
 \cK_{i,10}=\cP(3; p_i,p_7,p_8,p_9,p_{10}^2,p_{11}).
 $$
Given a point $p$ which is not a base point of $\cE$, there is a unique cubic curve of $\cK_{i,10}$ through $p$ and a unique sextic curve of $\cE$ through $p$. We set $I_{i,10}^{(3)}(p)=p'$, where $p'$ is the unique further intersection point of these two curves. This intersection is unique, since the intersection number of the cubic with the sextic is $3\times 6=18$, while the intersections at the points $p_i,p_7,p_8,p_9,p_{10},p_{11}$, and $p$ count as $1+2+2+2+6+3+1=17$. Involutions $I_{i,11}^{(3)}$ are defined similarly.
 \smallskip
  
  {\em 3)} Involutions $I_{i,j}^{(2)}$, $i,j\in\{7,8,9\}$, defined in terms of conics of the pencil
  $$
 \cC_{i,j}=\cP(2; p_i,p_j,p_{10},p_{11}).
 $$ 
Given a point $p$ which is not a base point of $\cE$, there is a unique conic of $\cC_{i,j}$ through $p$ and a unique sextic curve of $\cE$ through $p$. We set $I^{(2)}_{i,j}(p)=p'$, where $p'$ is the unique further intersection point of these two curves. This intersection is unique, since the intersection number of the conic with the sextic is $2\times 6=12$, while the intersections at the points $p_i,p_j,p_{10},p_{11}$, and $p$ count as $2+2+3+3+1=11$. 
 \end{definition}


\section{Quadratic Manin maps for special cubic pencils}
\label{section cubic}

In this section, we consider pencils of cubic curves,
$$
\cE=\cP(3;p_1,\ldots,p_9).
$$
Generically, a Manin involution for a cubic pencil is a birational map of degree 5 for which all base points of the pencil are singularities (indeterminacy points). Indeed, consider $I_{\cE,p_i}$. For any base point $p_j\neq p_i$, all curves $\cC_\lambda$ of the pencil pass through $p_i,p_j$, and have one further intersection point with the line $(p_ip_j)$. As a result, $I_{\cE,p_i}$ blows up any base point  $p_j$ $(j\neq i)$ to the line $(p_ip_j)$. For the same reason  $I_{\cE,p_j}$ blows down this line to $p_i$. Thus:
\begin{proposition} \label{prop cubic Manin transf}
For a cubic pencil, the Manin transformation $I_{\cE,p_i}\circ I_{\cE,p_j}$ for any two distinct base points $p_i$ and $p_j$ is regular at $p_i$ and maps it to $p_j$.
\end{proposition}
For a similar reason, some base points become regular points of Manin involutions if there are collinearities among them:
\begin{proposition} \label{prop cubic Manin transf collinear}
For a cubic pencil, if three distinct base points $p_i,p_j,p_k$ are collinear, then $I_{\cE,p_i}$ is regular at  $p_j$ and at $p_k$ and interchanges these two points.
\end{proposition}

We will say that the nine points $A_i$, $B_i$, $C_i$, $i=1,2,3$, form a \emph{Pascal configuration}, if the six points $A_1,A_2A_3,C_1,C_2,C_3$ lie on a conic, and 
\[
B_1=(A_2C_3)\cap(A_3C_2), \quad B_2=(A_3C_1)\cap(A_1C_3), \quad B_3=(A_1C_2)\cap(A_2C_1).
\]
By Pascal's theorem, the points $B_1,B_2,B_3$ are collinear.

\begin{figure}
\begin{center}
\includegraphics[width=0.7\textwidth]{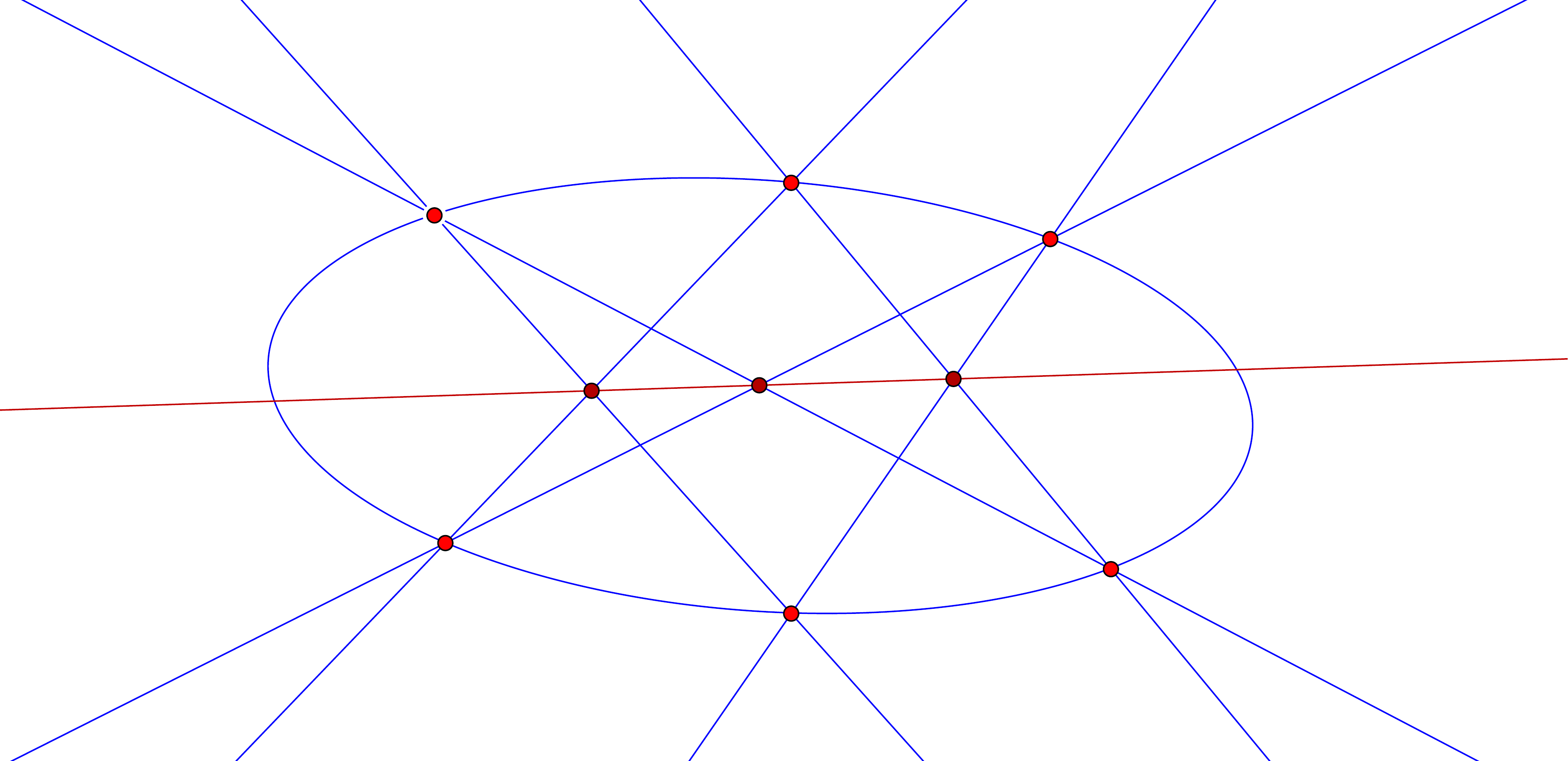}
\put(-239,123){\scalebox{0.8}{$A_1$}}
\put(-167,130){\scalebox{0.8}{$A_2$}}
\put(-117,118){\scalebox{0.8}{$A_3$}}
\put(-239,33){\scalebox{0.8}{$C_1$}}
\put(-167,18){\scalebox{0.8}{$C_2$}}
\put(-103,30){\scalebox{0.8}{$C_3$}}
\put(-132,67){\scalebox{0.8}{$B_1$}}
\put(-173,67){\scalebox{0.8}{$B_2$}}
\put(-209,67){\scalebox{0.8}{$B_3$}}
\caption{Pascal configuration of base points of a cubic pencil.}
\label{Fig Pascal}
\end{center}
\end{figure}
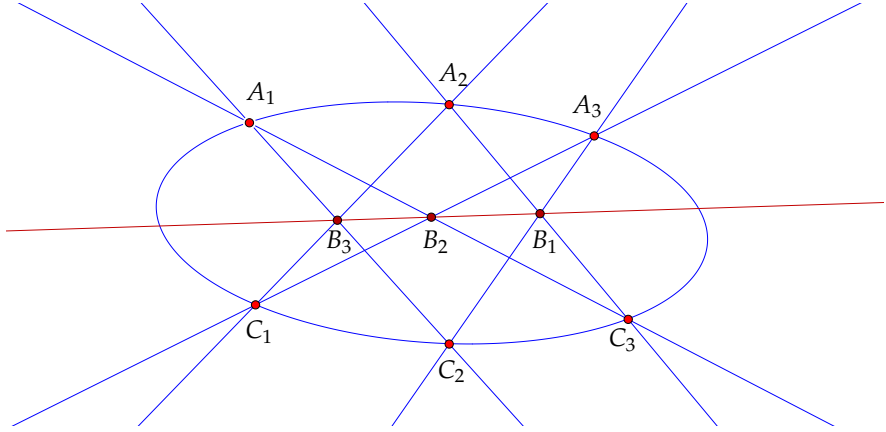
\begin{theorem}
Let the points $A_i$, $B_i$, $C_i$, $i=1,2,3$, form a Pascal configuration.
Consider the pencil $\cE$ of cubic curves with these base points. Then the map
\begin{eqnarray}
f & = & I_{\cE,A_1}\circ I_{\cE,B_1}\; = \; I_{\cE,B_1}\circ I_{\cE,C_1} \label{f Pascal 1}\\
 & = & I_{\cE,A_2}\circ I_{\cE,B_2}\; = \; I_{\cE,B_2}\circ I_{\cE,C_2} \label{f Pascal 2}\\
 & = & I_{\cE,A_3}\circ I_{\cE,B_3}\; = \; I_{\cE,B_3}\circ I_{\cE,C_3}  \label{f Pascal 3}
\end{eqnarray}
is a birational map of degree 2, with $\mathcal I(f)=\{C_1,C_2,C_3\}$ and $\mathcal I(f^{-1})=\{A_1,A_2,A_3\}$. It has the following singularity confinement patterns:
\begin{align}
& (C_2C_3)\to A_1\to B_1\to C_1\to (A_2A_3), \label{Pascal sing conf 1}\\
& (C_3C_1)\to A_2\to B_2\to C_2\to (A_3A_1),  \label{Pascal sing conf 2}\\
& (C_1C_2)\to A_3\to B_3\to C_3\to (A_1A_2).  \label{Pascal sing conf 3}
\end{align}
\end{theorem}
\begin{proof}
We start with the following property of the addition law on a nonsingular cubic curve $\cC$. Let $P_1,P_2,P_3,P_4\in\cC$, then 
$$
P_1-P_3=P_4-P_2 \;\; \Leftrightarrow \;\; P_1+P_2=P_3+P_4 \;\; \Leftrightarrow \;\; (P_1P_2)\cap (P_3P_4)\in \cC.
$$
Thus, on any cubic curve $\cC\in\cE$, we have the following relations:
\begin{align*}
& (A_1B_2)\cap (A_2B_1)=C_3\in\cC\;\; \Rightarrow \;\; A_1-B_1=A_2-B_2 \;\; \Rightarrow \;\; I_{\cE,A_1}\circ I_{\cE,B_1}=I_{\cE,A_2}\circ I_{\cE,B_2}, \\
& (B_1C_2)\cap (B_2C_1)=A_3\in\cC\;\; \Rightarrow \;\; B_1-C_1=B_2-C_2 \;\; \Rightarrow \;\; I_{\cE,B_1}\circ I_{\cE,C_1}=I_{\cE,B_2}\circ I_{\cE,C_2}, \\
& (A_1C_2)\cap (B_1B_2)=B_3\in\cC\;\; \Rightarrow \;\; A_1-B_1=B_2-C_2 \;\; \Rightarrow \;\; I_{\cE,A_1}\circ I_{\cE,B_1}=I_{\cE,B_2}\circ I_{\cE,C_2}.
\end{align*}
This proves the coincidence of all six representations in \eqref{f Pascal 1}--\eqref{f Pascal 3}. Now it follows from Proposition \ref{prop cubic Manin transf} that $f$ has only three indeterminacy points, $\cI(f)=\{C_1,C_2,C_3\}$, and similarly, $\ \cI(f^{-1})=\{A_1,A_2,A_3\}$. Moreover, Proposition \ref{prop cubic Manin transf} implies the relations in the middle part of the singularity confinement patterns \eqref{Pascal sing conf 1}--\eqref{Pascal sing conf 3}. The blow-up and blow-down relations are shown with the help of Proposition \ref{prop cubic Manin transf collinear} as follows: $f(C_3)=I_{\cE,A_2}\circ I_{\cE,B_2}(C_3)=I_{\cE,A_2}(A_1)=(A_1A_2)$.
\end{proof}

\begin{theorem}
For a pencil of cubic curves with the base points building a Pascal configuration, perform a linear projective transformation of $\bbP^2$ sending the Pascal line $\ell(B_1,B_2,B_3)$ to infinity. Let $(x,y)$ be the affine coordinates on the affine part $\bbC^2\subset\bbP^2$. In these coordinates, the map $f:(x,y)\mapsto (\t x,\t y)$ defined by \eqref{f Pascal 1}--\eqref{f Pascal 3} is characterized by the following property. There exist constants $a_1,\ldots,a_9\in\bbC$ such that $f$ admits a representation through two bilinear equations of motion of the form
\begin{equation}\label{Kahan 1}
\renewcommand{\arraystretch}{1.3}
\left\{\begin{array}{l}
\t x-x \; = \; a_2 x \t x+a_3(x\t y+\t x y)+a_4y\t y+a_6(x+\t x)+a_7(y+\t y)+a_9, \\
\t y-y \; = \; -a_1x\t x-a_2(x\t y+\t xy)-a_3y\t y-a_5(x+\t x)-a_6(y +\t y)-a_8. 
\end{array}\right.
\end{equation}
These equations serve as the Kahan discretization \cite{K, PPS2} of the Hamiltonian equations of motion
\begin{equation}\label{cubic Ham system}
\renewcommand{\arraystretch}{1.3}
\left\{\begin{array}{l}
\dot x \; = \; \partial H/\partial y\;=\; a_2x^2+2a_3xy+a_4y^2+2a_6x+2a_7y+a_9,   \\
\dot y \; = \; -\partial H/\partial x\;=\; -a_1x^2-2a_2xy-a_3y^2-2a_5x-2a_6y-a_8,  
\end{array}\right.
\end{equation}
for the Hamilton function
\begin{equation} \label{H cubic}
H(x,y)= \tfrac{1}{3}a_1 x^3 +a_2 x^2y +a_3 x y^2 +\tfrac{1}{3}a_4 y^3+a_5 x^2 +2a_6 xy +a_7 y^2 +a_8 x+ a_9y.
\end{equation}
\end{theorem}
\begin{proof} 
This is a result of a symbolic computation with MAPLE, presented in \cite{PSS}. 
\end{proof}

\section{Quadratic Manin maps for special pencils of the type $(4;8^1 2^2)$} 
\label{section quartic}

We describe the geometry of base points of a pencil of the type $(4;8^1 2^2)$ for which one can find compositions of Manin involutions which are quadratic Cremona maps.
\begin{itemize}
\item Let $p_2,p_3,p_6,p_7$ be four points of $\bbP^2$ in general position (no three of them collinear). 
\item Consider three intersection points of three pairs of opposite sides of the complete quadrangle with these vertices:
\begin{equation}
A=(p_2p_6)\cap (p_3p_7), \quad B=(p_2p_3)\cap (p_6p_7), \quad C=(p_2p_7)\cap(p_3p_6). 
\end{equation}
Consider the projective involutive automorphism $\sigma$ of $\bbP^2$ fixing the point $C$ and the line $\ell=(AB)$ (pointwise). The points of the pairs $(p_2,p_7)$ and $(p_3,p_6)$ correspond under $\sigma$.
\item Choose a point $p_9\in (p_3p_7)$, and define $p_{10}\in (p_2p_6)$ so that $p_9,p_{10}$ correspond under $\sigma$, or, in other words, so that the line $(p_9p_{10})$ passes through $C$.
\item  Let $\cC\in\cP(2;p_2,p_3,p_6,p_7)$ be any conic of the pencil through the specified four points. Define:
\begin{eqnarray*}
p_1 & = & \;\;{\rm the\;\;second\;\; intersection\;\,point\;\;of}\;\;\cC\;\;{\rm with}\;\;  (p_{10}p_7),\\
p_4 & = & \;\;{\rm the\;\;second\;\; intersection\;\,point\;\;of}\;\;\cC\;\;{\rm with}\;\;  (p_9p_6),\\
p_5 & = & \;\;{\rm the\;\;second\;\; intersection\;\,point\;\;of}\;\;\cC\;\;{\rm with}\;\;  (p_{10}p_3), \\
p_8 & = & \;\;{\rm the\;\;second\;\; intersection\;\,point\;\;of}\;\;\cC\;\;{\rm with}\;\;  (p_9p_2).
\end{eqnarray*}
Recall that $A,B,C$ are vertices of a self-polar triangle for $\cC$. In particular, the projective involution $\sigma$ leaves $\cC$ invariant. The points of the pairs $(p_1,p_8)$ and $(p_4,p_5)$ correspond under $\sigma$.
\end{itemize}

We will call the pencil $\cE=\cP(p_1,\ldots,p_8,p_9^2,p_{10}^2)$ a \emph{projectively symmetric quartic pencil with two double points}.

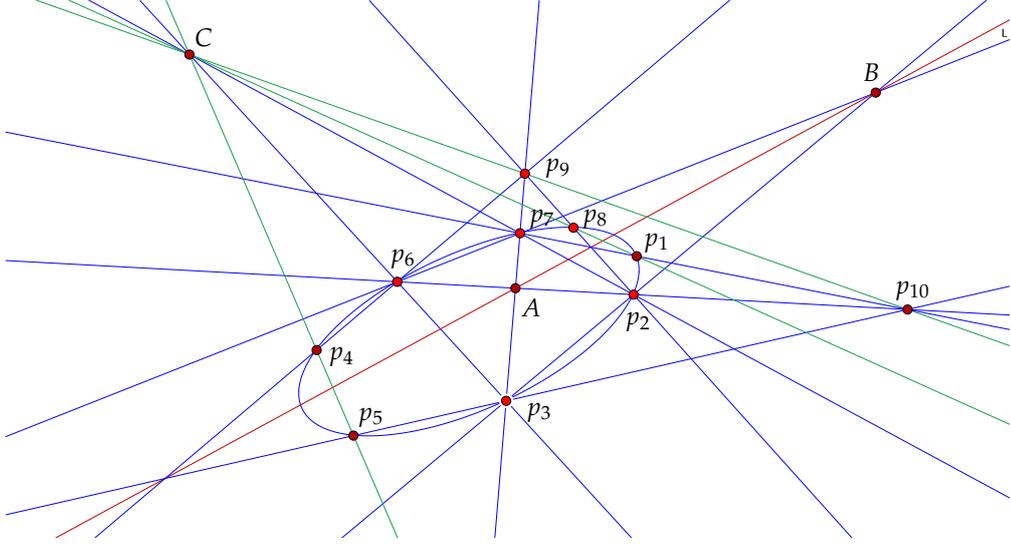
\begin{figure}
\begin{center}
\includegraphics[width=0.8\textwidth]{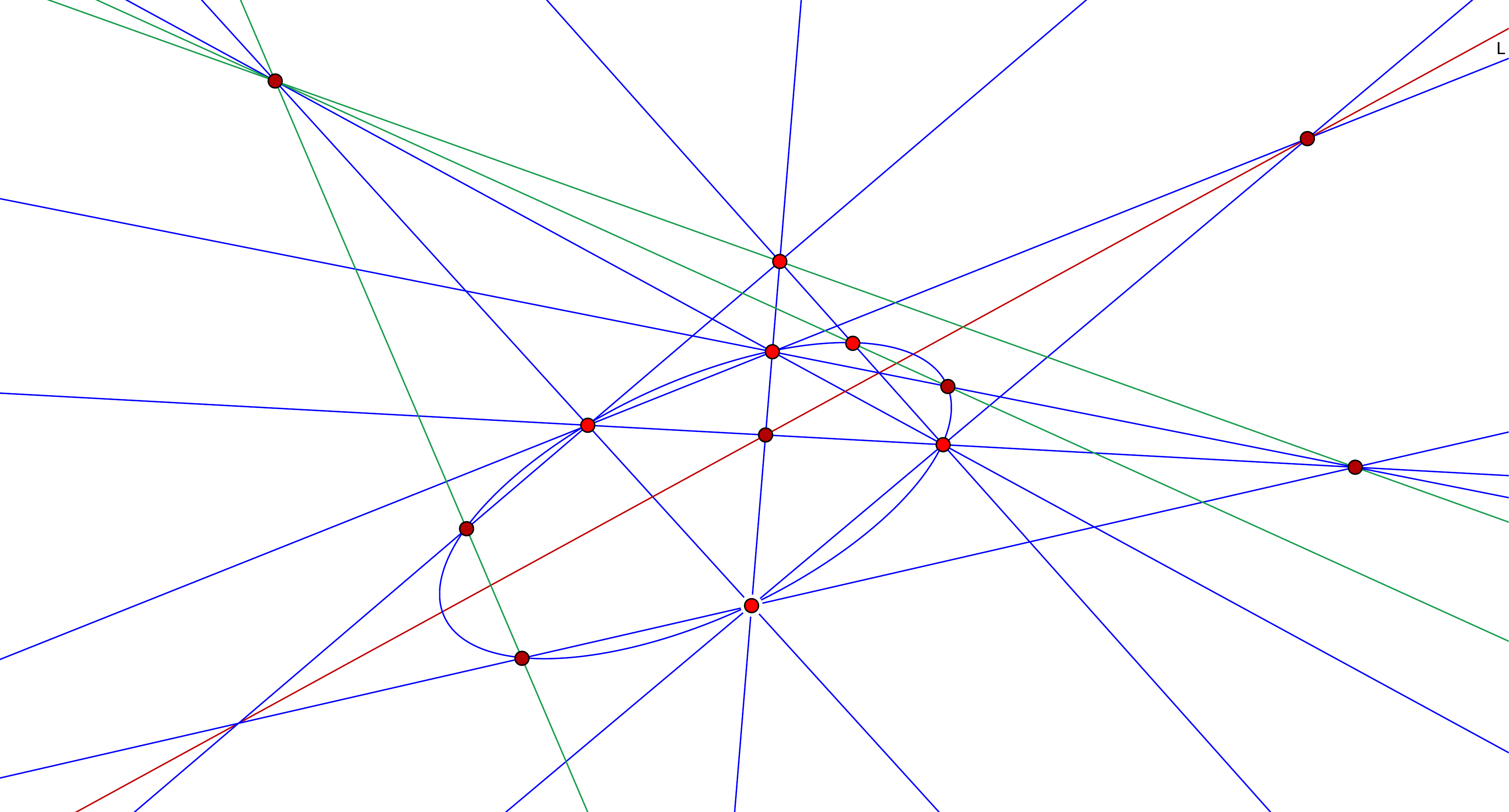}
\put(-305,185){\scalebox{0.8}{$C$}}
\put(-55,172){\scalebox{0.8}{$B$}}
\put(-183,83){\scalebox{0.8}{$A$}}
\put(-43,92){\scalebox{0.8}{$p_{10}$}}
\put(-174,138){\scalebox{0.8}{$p_9$}}
\put(-160,119){\scalebox{0.8}{$p_8$}}
\put(-180,119){\scalebox{0.8}{$p_7$}}
\put(-232,104){\scalebox{0.8}{$p_6$}}
\put(-244,44){\scalebox{0.8}{$p_5$}}
\put(-255,67){\scalebox{0.8}{$p_4$}}
\put(-181,46){\scalebox{0.8}{$p_3$}}
\put(-144,80){\scalebox{0.8}{$p_2$}}
\put(-137,110){\scalebox{0.8}{$p_1$}}
\caption{Geometry of base points of a special quartic pencil $\cP(4;p_1,\ldots,p_8,p_9^2,p_{10}^2)$.}
\label{Fig quartic}
\end{center}
\end{figure}

\begin{theorem}\label{Th f quartic}
Let $\cE=\cP(p_1,\ldots,p_8,p_9^2,p_{10}^2)$ be a projectively symmetric quartic pencil with two double points. Then:

a) The projective involution $\sigma$ can be represented as 
\begin{equation}
\sigma=I^{(2)}_{1,8}=I^{(2)}_{2,7}=I^{(2)}_{3,6}=I^{(2)}_{4,5}.
\end{equation}

b) The map
\begin{eqnarray}
f  & =  & I_{i,k}^{(2)}\circ I_{j,k}^{(2)}  \label{f quartic}\\
   & = & I^{(1)}_9\circ \sigma \;\; = \;\; \sigma\circ I^{(1)}_{10} \label{f quartic 2}
\end{eqnarray}
with $(i,j)\in\{(1,2),(2,3),(3,4),(5,6),(6,7),(7,8)\}$ and $k\in\{1,\ldots,8\}\setminus\{i,j\}$,
is a birational map of degree 2, with $\mathcal I(f)=\{p_4,p_8,p_{10}\}$ and $\mathcal I(f^{-1})=\{p_1,p_5,p_9\}$. It has the following singularity confinement patterns:
\begin{align}
& (p_8p_{10})\to p_1\to p_2\to p_3\to p_4 \to  (p_5p_9), \label{quart sing conf 1}\\
& (p_4p_{10})\to p_5\to p_6\to p_7\to p_8 \to (p_1p_9),  \label{quart sing conf 2}\\
& (p_4p_8)\to p_9\to p_{10}\to  (p_1p_5).  \label{quart sing conf 3}
\end{align}

c) We have:
\begin{equation}\label{f^2}
f^2=I^{(1)}_9 \circ I^{(1)}_{10}.
\end{equation}
\end{theorem}
\begin{proof}
We start with a geometric interpretation of the addition law on a generic curve $\cC\in\cE$. Recall that the pencil $\cE$ can be reduced to a pencil of cubic curves by means of the quadratic Cremona map $\phi$ based at $p_k,p_9,p_{10}$ for some $k=1,\ldots,8$. Lines in the target plane $\bbP^2_2$, where the cubic pencil is considered, correspond in the source plane $\bbP^2_1$ of the pencil $\cE$ to conics through $p_k,p_9,p_{10}$. Now let $p,q,r,s\in\cC$, then, assuming that neither of the points $p_k,p_9,p_{10}$ is among $p,q,r,s$, we have: 
$$
p-r=s-q \;\; \Leftrightarrow \;\; p+q=r+s \;\; \Leftrightarrow \;\; \big(\phi(p)\phi(q)\big)\cap \big(\phi(r)\phi(s)\big)\in \phi(\cC).
$$
The geometry of the pencil $\cE$ ensures the existence of a large number of quadruples of base points which, together with $p_9,p_{10}$, lie on a conic. Namely, the following sextuples are conconical: 
\begin{align}
& (p_1,p_2,p_7,p_8,p_9,p_{10})\;\; {\rm because} \;\; p_1\leftrightarrow p_8,\;\; p_2\leftrightarrow p_7 \;\;{\rm under}\;\; \sigma, \label{conic 1}\\
& (p_1,p_3,p_6,p_8,p_9,p_{10})\;\; {\rm because} \;\; p_1\leftrightarrow p_8,\;\; p_3\leftrightarrow p_6 \;\;{\rm under}\;\; \sigma, \label{conic 2}\\
& (p_1,p_4,p_5,p_8,p_9,p_{10})\;\; {\rm because} \;\; p_1\leftrightarrow p_8,\;\; p_4\leftrightarrow p_5 \;\;{\rm under}\;\; \sigma, \label{conic 3}\\
& (p_2,p_3,p_6,p_7,p_9,p_{10})\;\; {\rm because} \;\; p_2\leftrightarrow p_7,\;\; p_3\leftrightarrow p_6 \;\;{\rm under}\;\; \sigma, \label{conic 4}\\
& (p_2,p_4,p_5,p_7,p_9,p_{10})\;\; {\rm because} \;\; p_2\leftrightarrow p_7,\;\; p_4\leftrightarrow p_5 \;\;{\rm under}\;\; \sigma, \label{conic 5}\\
& (p_3,p_4,p_5,p_6,p_9,p_{10})\;\; {\rm because} \;\; p_3\leftrightarrow p_6,\;\; p_4\leftrightarrow p_5 \;\;{\rm under}\;\; \sigma. \label{conic 6}
\end{align}
The sextuples \eqref{conic 1}, \eqref{conic 4} and \eqref{conic 6} lie on reducible conics $\ell(p_1,p_7,p_{10})\cup \ell(p_2,p_8,p_9)$, \linebreak $\ell(p_2,p_6,p_{10})\cup \ell(p_3,p_7,p_9)$ and $\ell(p_3,p_5,p_{10})\cup \ell(p_4,p_6,p_9)$, respectively. One has, additionally, two more sextuples lying on reducible conics:
\begin{align}
& (p_1,p_4,p_6,p_7,p_9,p_{10})\;\;- \;\; {\rm on\;\;a\;\;reducible\;\;conic} \;\; \ell(p_1,p_7,p_{10})\cup \ell(p_4,p_6,p_9), \label{conic 7}\\
& (p_2,p_3,p_5,p_8,p_9,p_{10})\;\;- \;\; {\rm on\;\;a\;\;reducible\;\;conic} \;\; \ell(p_3,p_5,p_{10})\cup \ell(p_2,p_8,p_9). \label{conic 8}
\end{align}
\begin{itemize}
\item 
From \eqref{conic 2}, \eqref{conic 8}, \eqref{conic 7}, \eqref{conic 5} there follows:
$$
\left\{ \begin{array}{l}
C(p_1,p_6,p_3,p_9,p_{10})\cap C(p_2,p_5,p_3,p_9,p_{10})\ni p_8,\\
C(p_1,p_6,p_4,p_9,p_{10})\cap C(p_2,p_5,p_4,p_9,p_{10})\ni p_7, \\
C(p_1,p_6,p_7,p_9,p_{10})\cap C(p_2,p_5,p_7,p_9,p_{10})\ni p_4, \\
C(p_1,p_6,p_8,p_9,p_{10})\cap C(p_2,p_5,p_8,p_9,p_{10})\ni p_3.
\end{array}\right.
$$
We explain how these relations are used, taking the first one as example. The intersection $C(p_1,p_6,p_3,p_9,p_{10})\cap C(p_2,p_5,p_3,p_9,p_{10})$ consists of $p_3$, $p_9$, $p_{10}$, and $p_8$. Upon the quadratic Cremona map $\phi$ based at $p_3,p_9,p_{10}$, this means that the lines $(q_1q_6)$ and $(q_2q_5)$ intersect at $q_8$, where $q_i=\phi(p_i)$ (the blow-ups of other three intersection points do not belong to the proper image of the conics).  The point $q_8$ is one of the base points of the cubic pencil $\phi(\cE)$.
Thus, the four relations above imply
\begin{equation}
I^{(2)}_{1,k}\circ I^{(2)}_{2,k}=I^{(2)}_{5,k}\circ I^{(2)}_{6,k}, \quad k=3,4,7,8.
\end{equation}

\item From \eqref{conic 6}, \eqref{conic 4},  \eqref{conic 8} there follows:
$$
\left\{ \begin{array}{l}
C(p_2,p_6,p_1,p_9,p_{10})\cap C(p_3,p_5,p_1,p_9,p_{10})\supset (p_1p_9), \\
C(p_2,p_6,p_4,p_9,p_{10})\cap C(p_3,p_5,p_4,p_9,p_{10})\supset (p_4p_9), \\
C(p_2,p_6,p_7,p_9,p_{10})\cap C(p_3,p_5,p_7,p_9,p_{10})\supset (p_7p_9), \\
C(p_2,p_6,p_8,p_9,p_{10})\cap C(p_3,p_5,p_8,p_9,p_{10})\supset (p_8p_9).
\end{array}\right.
$$
Again, we explain how these relations are used, taking the first one as example. The intersection $C(p_2,p_6,p_1,p_9,p_{10})\cap C(p_3,p_5,p_1,p_9,p_{10})$ consists of the point $p_{10}$ and the line $(p_1p_9)$. Upon the quadratic Cremona map $\phi$ based at $p_1,p_9,p_{10}$,
the point $p_{10}$ is blown up to a line which does not belong to the proper image of the conics, while the line $(p_1p_9)$ is blown down to the point $q_{10}$ through which the proper images of the both conics pass. Thus, the lines $(q_2q_6)$ and $(q_3q_5)$ intersect at  $q_{10}$, which is a base point of the pencil $\phi(\cE)$. Summarizing, the four relations above imply
\begin{equation}
I^{(2)}_{2,k}\circ I^{(2)}_{3,k}=I^{(2)}_{5,k}\circ I^{(2)}_{6,k}, \quad k=1,4,7,8.
\end{equation}

\item 
From \eqref{conic 2} , \eqref{conic 3}, \eqref{conic 4}, \eqref{conic 5}, there follows:
$$
\left\{ \begin{array}{l}
C(p_3,p_6,p_1,p_9,p_{10})\cap C(p_4,p_5,p_1,p_9,p_{10})\ni p_8, \\
C(p_3,p_6,p_2,p_9,p_{10})\cap C(p_4,p_5,p_2,p_9,p_{10})\ni p_7,\\
C(p_3,p_6,p_7,p_9,p_{10})\cap C(p_4,p_5,p_7,p_9,p_{10})\ni p_2,\\
C(p_3,p_6,p_8,p_9,p_{10})\cap C(p_4,p_5,p_8,p_9,p_{10})\ni p_1.
\end{array}\right.
$$
Exactly as before, these four relations imply
\begin{equation}
I^{(2)}_{3,k}\circ I^{(2)}_{4,k}=I^{(2)}_{5,k}\circ I^{(2)}_{6,k}, \quad k=1,2,7,8.
\end{equation}

\item
In exactly the same way we prove that
\begin{equation}
I^{(2)}_{1,k}\circ I^{(2)}_{2,k}=I^{(2)}_{6,k}\circ I^{(2)}_{7,k}, \quad k=3,4,5,8.
\end{equation}
and
\begin{equation}
I^{(2)}_{1,k}\circ I^{(2)}_{2,k}=I^{(2)}_{7,k}\circ I^{(2)}_{8,k}, \quad k=3,4,5,6.
\end{equation}
This completes the proof of coincidence of all representations \eqref{f quartic}, as well as the middle part of the singularity confinement patterns \eqref{quart sing conf 1}, \eqref{quart sing conf 2}.

\item One sees immediately that $I^{(1)}_9\circ \sigma$ is a shift with respect to the addition law on the curves of $\cE$, sending $p_1\to p_2\to p_3\to p_4$, while $\sigma\circ I^{(1)}_{10}$ is a shift sending $p_5\to p_6\to p_7\to p_8$. Therefore, these shifts must coincide with $f$. This proves \eqref{f quartic 2} and the middle part of the  singularity confinement pattern \eqref{quart sing conf 3}.

\item Collecting all the results, we see that $\cI(f)=\{p_4,p_8,p_{10}\}$ and $\cI(f^{-1})=\{p_1,p_5,p_{9}\}$, so that $f$ must be a quadratic Cremona map. 
\item It remains to show the blow-up and blow-down relations in the singularity confinement patterns \eqref{quart sing conf 1}--\eqref{quart sing conf 3}. 
To see the blow-down relations on the left, we use the representation $f=I_9^{(1)} \circ \sigma$. The involution $\sigma$ is a projective automorphism and has no singularities, so it suffices to study the blowing down patterns of $I_9^{(1)}$.  By definition of the map $I_9^{(1)}$, it is clear that it blows down the line $(p_1p_9)$ to the point $p_1$, and blows down the line $(p_5p_9)$ to the point $p_5$. Since $f$ is a quadratic Cremona map, the same holds true for the involution $I_9^{(1)}$; there follows that $I_9^{(1)}$ must blow up $p_9$ to the line $(p_1p_5)$, which finishes the proof. For the blow-up relations on the right part of \eqref{quart sing conf 1}--\eqref{quart sing conf 3}, we use $f=\sigma \circ I_{10}^{(1)}$ in a similar manner.
\end{itemize}
\end{proof}

We now turn to canonical forms of projectively symmetric quartic pencils with two double points, which can be achieved by projective automorphisms of $\bbP^2$. The most popular one corresponds to the choice $p_9=[0:1:0]$, $p_{10}=[1:0:0]$, so that the quartic curves become \emph{biquadratic} ones. Denote the inhomogeneous coordinates on the affine part $\bbC^2\subset\bbP^2$ by $(u,v)$. We can arrange $p_2=(1,b)$, $p_7=(b,1)$, $p_3=(a,-1)$, $p_6=(-1,a)$, so that $\ell=\{u-v=0\}$, $C=(p_2p_7)\cap(p_3p_6)=[-1:1:0]$, and $\sigma$ is the Euclidean reflection at the line $\ell$, 
$$
\sigma(u,v)=(v,u).
$$
The pencil $\cE$ of biquadratics reads
\begin{equation}\label{pencil QRT root 2}
\alpha(\alpha+1)(u^2+v^2-1)-(\alpha+1)uv+\beta(u+v)-\beta^2-\lambda(u^2-1)(v^2-1)=0,
\end{equation}
and is symmetric under $\sigma$. Involutions $I^{(1)}_9$ and $I_{10}^{(1)}$ are nothing but the standard vertical and horizontal QRT switches for this pencil, and 
the map $f=I^{(1)}_9\circ \sigma=\sigma\circ I_{10}^{(1)}$ of  Theorem \ref{Th f quartic} is given by 
\begin{equation}\label{QRT root 2}
\renewcommand{\arraystretch}{1.4}
f:(u,v)\mapsto (\t u,\t v), \quad 
\left\{\begin{array}{l} \t u=v, \\ \t v=\dfrac{\alpha uv+\beta u-1}{u-\alpha v-\beta}. \end{array}\right.
\end{equation}
It is the ``QRT root'' of $f^2=I^{(1)}_9\circ I_{10}^{(1)}$.

To arrive at another canonical form of projectively symmetric quartic pencils with two double points, we perform a linear projective change of variables in $\bbP^2$, given in the non-homogeneous coordinates by
\begin{equation}\label{change 2}
u=\frac{1+\beta x+y}{x}, \quad v=\frac{1+\beta x-y}{x}.
\end{equation}
Upon substitution \eqref{change 2} and some straightforward simplifications, we come to the following system (cf. \cite{PSZ}):
\begin{equation}\label{QRT to Kahan 2}
\renewcommand{\arraystretch}{1.3}
\left\{\begin{array}{l}
\t x-x=x\t y+\t xy, \\
\t y -y=(1-2\alpha)-2\alpha\beta(x+\t x)+\big(1-\beta^2(1+2\alpha)\big)x\t x-(1+2\alpha)y\t y.
\end{array}\right.
\end{equation}
In order to give an intrinsic geometric characterization of this canonical form, we will need the following observation.

\begin{proposition} \label{prop double Pascal}
The following five intersection points are collinear:
$$
(p_1p_8)\cap (p_2p_7), \;\; (p_1p_5)\cap (p_3p_7), \;\; (p_2p_5)\cap (p_3p_8), \;\; (p_1p_6)\cap (p_4p_7), \;\;  (p_2p_6)\cap (p_4p_8).
$$
\end{proposition}
\begin{proof}
The triple of intersection points
$$
(p_1p_8)\cap (p_2p_7), \;\; (p_1p_5)\cap (p_3p_7), \;\; (p_2p_5)\cap (p_3p_8)
$$
lies on the Pascal line for the hexagon $(p_1,p_5,p_2,p_7,p_3,p_8)$, while the triple of intersection points
$$
(p_1p_8)\cap (p_2p_7),  \;\; (p_1p_6)\cap (p_4p_7), \;\;  (p_2p_6)\cap (p_4p_8)
$$
lies on the Pascal line for the hexagon $(p_1,p_6,p_2,p_7,p_4,p_8)$. These hexagons correspond under $\sigma$, therefore this holds true also for their Pascal lines. Moreover, the Pascal lines share the point $(p_1p_8)\cap (p_2p_7)=C$, therefore they must coincide.
\end{proof}

We will call the line containing the five intersection points from Proposition \ref{prop double Pascal} the \emph{double Pascal line}.

\begin{theorem}
For a projectively symmetric pencil of quartic curves with two double points, perform a linear projective transformation of $\ \bbP^2$ sending the double Pascal line to infinity. By a subsequent affine change of coordinates $(x,y)$ on the affine part $\bbC^2\subset\bbP^2$, arrange that $\ell$ coincides with the axis $y=0$, $p_9=(0,-1)$, and $p_{10}=(0,1)$. In these coordinates, the map $f:(x,y)\mapsto (\t x,\t y)$ defined by \eqref{f quartic}--\eqref{f quartic 2} is characterized by the following property.  There exist $a_0,\ldots,a_3\in\bbC$ with $a_0+a_3=2$ such that $f$ admits a representation through two bilinear equations of motion of the form 
\begin{equation}\label{Kahan 2}
\renewcommand{\arraystretch}{1.2}
\left\{\begin{array}{l}
\t x-x=x\t y+\t xy, \\
\t y -y=a_0-a_1(x+\t x)-a_2x\t x-a_3y\t y.
\end{array}\right.
\end{equation}
\end{theorem}
\begin{proof}
A symbolic computation with MAPLE.
\end{proof}


\section{Quadratic Manin maps for special pencils of the type $(6;6^1 3^2 2^3)$}
\label{section sextic}

In this section, we consider the following example considered in detail in \cite{PPS2, PZ, CMOQ4, Z}. Let $f:\bbP^2\to\bbP^2$ be the birational map given in the non-homogeneous coordinates $\X=(x,y)$ on the affine part $\bbC^2\subset\bbP^2$ by two bilinear relations between $(x,y)$ and $(\t x,\t y)=f(x,y)$:
\begin{align} \label{nahm HK}
\t \X-\X \; = \;\; & \gamma_1(\ell_2(\X)\ell_3(\t \X)+\ell_2(\t \X)\ell_3(\X))J\nabla\ell_1 \nonumber \\
+ & \gamma_2(\ell_1(\X)\ell_3(\t \X)+\ell_1(\t \X)\ell_3(\X))J\nabla\ell_2 \nonumber \\
+ & \gamma_3(\ell_1(\X)\ell_2(\t \X)+\ell_1(\t \X)\ell_2(\X) )J\nabla\ell_3,
\end{align}
where $(\gamma_1,\gamma_2,\gamma_3)=(1,2,3)$, $\ell_i(\X)=a_ix+b_iy$ are linear forms with $a_i,b_i\in\bbC$, and $J=\begin{pmatrix}
0&1\\ -1&0 \end{pmatrix}$. This map is the Kahan discretization of the quadratic flow
\begin{equation} \label{nahm}
\dot \X =\gamma_1\ell_2(\X)\ell_3(\X)J\nabla\ell_1+\gamma_2\ell_1(\X)\ell_3(\X)J\nabla\ell_2+\gamma_3\ell_1(\X)\ell_2(\X)J\nabla\ell_3,
\end{equation}
which can be put as 
\begin{equation} \label{nahm}
\dot \X =\frac{1}{(\ell_1(\X))^{\gamma_1-1}(\ell_2(\X))^{\gamma_2-1}(\ell_3(\X))^{\gamma_3-1}}J\nabla H_0(\X),
\end{equation}
where
\begin{equation} \label{nahm H0}
H_0(\X) =(\ell_1(\X))^{\gamma_1}(\ell_2(\X))^{\gamma_2}(\ell_3(\X))^{\gamma_3}.
\end{equation}
Integrability of the Kahan discretization \eqref{nahm HK} was demonstrated in \cite{PPS2, PZ, CMOQ4} for $(\gamma_1,\gamma_2,\gamma_3)=(1,1,1), (1,1,2)$, and $(1,2,3)$. Sections \ref{section cubic} and \ref{section quartic} deal with generalizations of the cases $(\gamma_1,\gamma_2,\gamma_3)=(1,1,1), (1,1,2)$, the present one deals with $(\gamma_1,\gamma_2,\gamma_3)=(1,2,3)$.

As shown in \cite{PPS2, PZ, CMOQ4}, the map $f$ defined by \eqref{nahm HK} admits an integral of motion:
\begin{equation}
\label{123_integral}
H(\X)=\dfrac{H_0(\X)}{L_+(\X)L_-(\X)M_+(\X)M_-(\X)Q(\X)},
\end{equation}
where
\begin{eqnarray*}
L_{\pm}(\X) & = & 1\pm 3d_{31}\ell_2(\X), \\
M_\pm (\X) & = & 1 \pm\left(3d_{23}\ell_1(\X)-d_{12}\ell_3(\X)\right),\\
Q(\X) & = & 1-\big(9d_{31}^2\ell_2^2(\X)+16d_{12}^2\ell_3^2(\X)\big),
\end{eqnarray*}
with
\begin{equation*}
d_{ij}=a_ib_j-a_jb_i.
\end{equation*}
Thus, the phase space of $f$ is foliated by the pencil of invariant curves 
\begin{equation}\label{sextic pencil 123}
C_\lambda=\big\{ H_0(\X)-\lambda L_+(\X)L_-(\X)M_+(\X)M_-(\X)Q(\X)=0\big\}.
\end{equation}
The pencil has $\deg=6$ and contains two reducible curves: $C_0=\{H_0(x,y)=0\}$, consisting of the lines $\{\ell_i(x,y)=0\}$, $i=1,2,3$, with multiplicities $1,2,3$, and  $C_\infty$,
consisting of the conic $Q(x,y)=0$ and the four lines $L_{\pm}(x,y)=0$, $M_{\pm}(x,y)=0$.
All curves $C_{\lambda}$ pass through the set of base points which is defined as $C_0\cap C_{\infty}$. 
One easily computes the 11 (distinct) base points of the pencil. They are given by:
\begin{itemize}
\item six base points of multiplicity $1$ on the line $\ell_1=0$:
$$
p_1=\Big(-\frac{b_1}{5d_{12}d_{31}},\frac{a_1}{5d_{12}d_{31}}\Big),\quad 
p_2=\Big(-\frac{b_1}{3d_{12}d_{31}},\frac{a_1}{3d_{12}d_{31}}\Big),\quad
p_3=\Big(-\frac{b_1}{d_{12}d_{31}},\frac{a_1}{d_{12}d_{31}}\Big),
$$
$$
p_4=-p_3,\quad 
p_5=-p_2,\quad
p_6=-p_1;
$$

\item three base points of multiplicity $2$ on the line $\ell_2=0$:
$$
p_7=\Big(-\frac{b_2}{4d_{12}d_{23}},\frac{a_2}{4d_{12}d_{23}}\Big),\quad p_8=[b_2:-a_2:0], \quad p_9=-p_7;
$$

\item and two base points of multiplicity $3$ on the line $\ell_3=0$:
$$
p_{10}=\Big(-\frac{b_3}{3d_{23}d_{31}},\frac{a_3}{3d_{23}d_{31}}\Big), \quad p_{11}=-p_{10}. 
$$
\end{itemize}
The map $f$ is a quadratic Cremona map with the indeterminacy points $\cI(f)=\{p_6,p_9,p_{11}\}$ and $\cI(f^{-1})=\{p_1,p_7,p_{10}\}$, and has the following singularity confinement patterns:
\begin{align*}
& (p_9p_{11})\to p_1 \to p_2 \to p_3 \to p_4 \to p_5 \to p_6 \to (p_7p_{10}),\\
& (p_6p_{11}) \to p_7 \to p_8 \to p_9 \to (p_1p_{10}),\\
& (p_6p_9) \to p_{10} \to p_{11} \to (p_1p_7).
\end{align*}

\begin{figure}[htbp]
\includegraphics[scale=0.7]{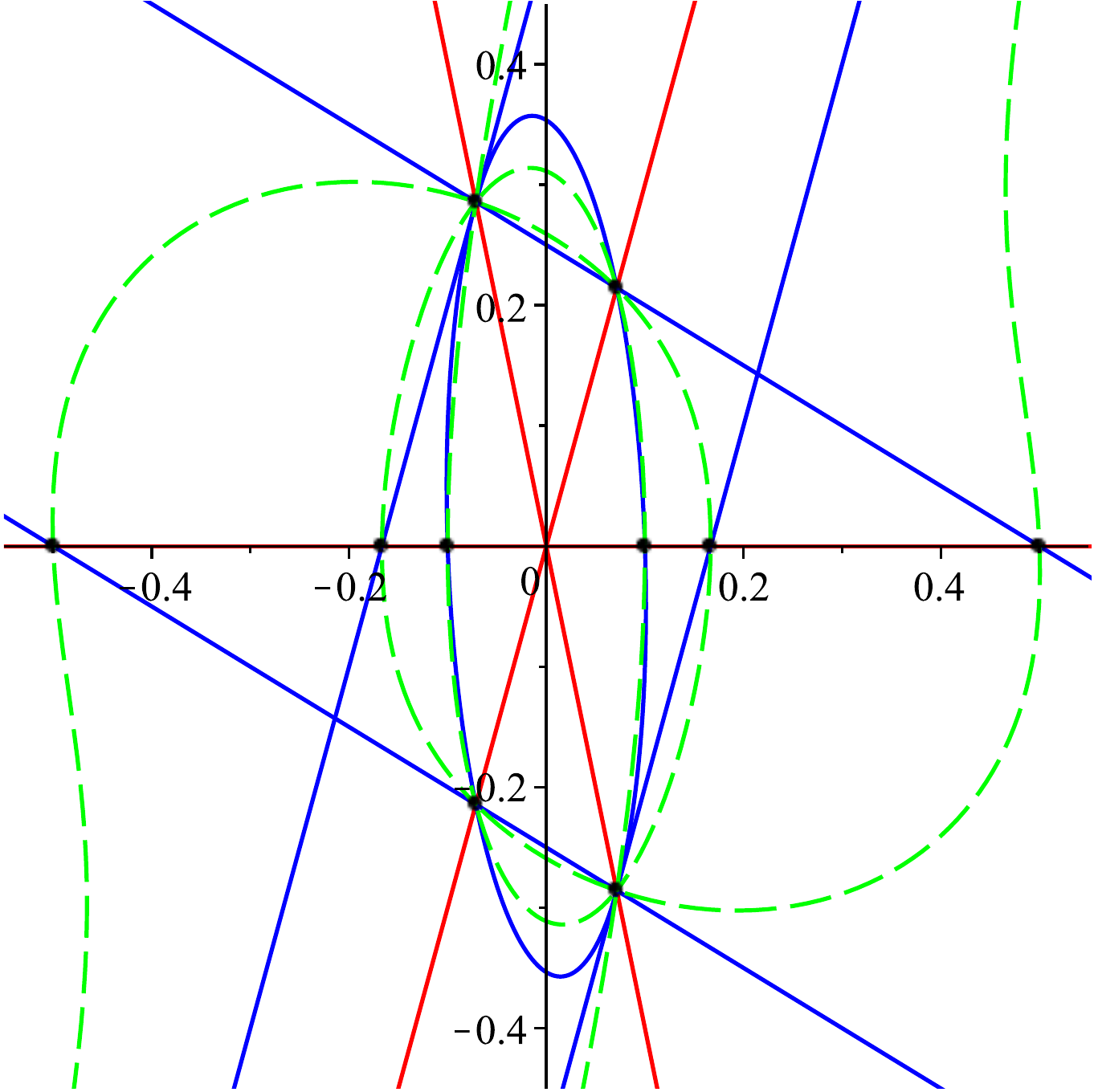}
\put(-5,120){$x$}
\put(-125,255){$y$}
\put(-127,145){$p_1$}
\put(-90,145){$p_2$}
\put(-10,145){$p_3$}
\put(-160,145){$p_6$}
\put(-195,145){$p_5$}
\put(-260,145){$p_4$}
\put(-175,60){$p_7$}
\put(-115,205){$p_9$}
\put(-180,215){$p_{10}$}
\put(-115,55){$p_{11}$}
\caption{The curves $C_0$, $C_{\infty}$, $C_{-0.002}$ of the sextic pencil (in red, blue and green, respectively) for $\ell_1(\X)=y/6$, $\ell_2(\X)=3x-y$, $\ell_3(\X)=4x+y$.}
\end{figure}
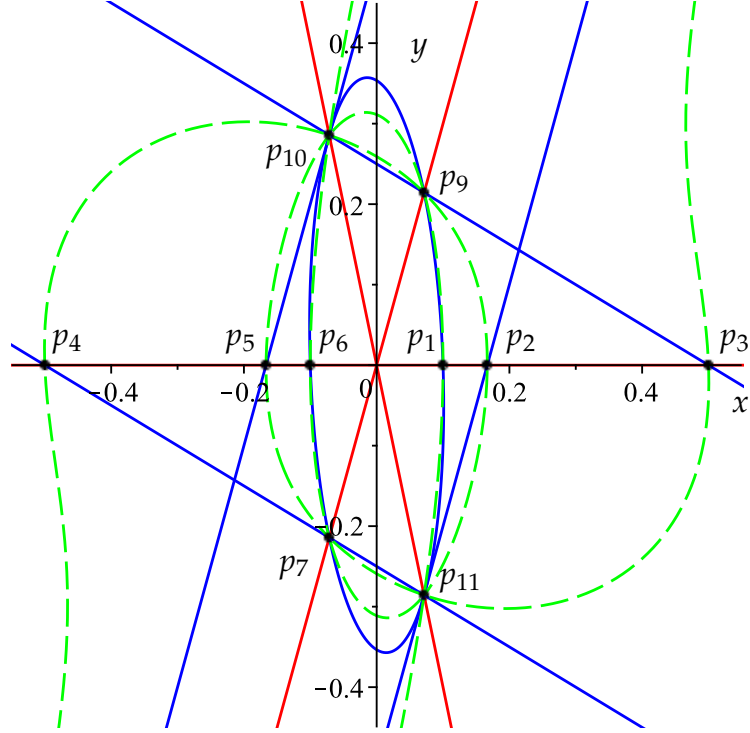

Unlike the previous two sections, we will not derive here general geometric conditions for the base points of the pencil which ensure that certain compositions of Manin involutions are quadratic Cremona maps. Rather, we will give here the corresponding statements for the pencil \eqref{sextic pencil 123}.
\begin{theorem}
The map $f$ can be represented as compositions of the Manin involutions in the following ways:
$$
f=I^{(4)}_{i,k,m}\circ I^{(4)}_{j,k,m}=I^{(3)}_{i,n}\circ I^{(3)}_{j,n}
$$
for any $(i,j)\in\{(1,2),(2,3),(3,4),(4,5),(5,6)\}$, $k\in\{1,\ldots,6\}\setminus \{i,j\}$, and $m\in\{7,8,9\}$, $n\in\{10,11\}$.
\end{theorem}
\begin{proof}
Symbolic computation with MAPLE.
\end{proof}

\section{Conclusions} The  contribution of this paper is two-fold:
\begin{itemize}
\item Finding geometric description of Manin involutions for elliptic pencil consisting of curves of higher degree, birationally equivalent to cubic pencils (Halphen pencils of index 1).

\item Characterizing special geometry of base points ensuring that certain compositions of Manin involutions are integrable maps of low degree (quadratic Cremona maps). As particular cases, we identify some integrable Kahan discretizations as compositions of Manin involutions.
\end{itemize}

It should be mentioned that both issues can and should be studied also for Halphen pencils of index $m>1$. For instance, for a Halphen pencil of index 2, $\cP(6;p_1^2,...,p_9^2)$, one can propose the following construction of involutions $I_{p_i}$, $i=1,\ldots,9$. For any $p\in\bbP^2$ different from the base points, consider the cubic curve through $p$ and $p_2,\ldots,p_9$. Its intersection number with the Halphen's curve of degree 6 through $p$ is $3\times 6=18$. The intersections with eight base points $p_2,\ldots,p_9$ and with $p$ count as $8\times 2+1=17$, so there is exactly one remaining intersection point $p'$. We declare $p'=I_{p_1}(p)$. One can see that the so defined involutions $I_{p_i}:\bbP^2\to\bbP^2$ are Bertini involutions (of degree 17 in the generic situation). However, one can show that the birational map (6) of degree 3 from \cite{CT1} is a composition of two such involutions, due to very special geometry of the base points, involving infintely near ones. We hope to be able to identify in the future work further low degree integrable birational maps as compositions of fundamental involutions defined by elliptic pencils.

\section{Acknowledgement}
This research is supported by the DFG Collaborative Research Center TRR 109 ``Discretization in Geometry and Dynamics''.

{}

\begin{thebibliography}{}

\bibitem{BB}
L. Bayle, A. Beauville.
{\em Birational involutions of $\bf P^2$},
Asian J. Math. {\bf 4} (2000), 11--18.


\bibitem{Bertini}
E. Bertini.
{\em Ricerche sulle trasformazioni univoche involutorie nel piano}, 
Annali di Mat., {\bf 8} (1877), 244--286.

\bibitem{CT1}
A.S. Carstea, T. Takenawa.
{\em A classification of two-dimensional integrable mappings and rational elliptic surfaces},
J. Phys. A: Math. Theor. {\bf 45} (2012), 155206 (15 pp).

\bibitem{CT2}
A.S. Carstea, T. Takenawa.
{\em A note on minimization of rational surfaces obtained from birational dynamical systems},
J. Nonlin. Math. Phys., {\bf 20} (2013), Suppl. 1, 17--33.


\bibitem{CMOQ1}
E. Celledoni, R.I. McLachlan, B. Owren, G.R.W. Quispel.
{\em Geometric properties of Kahan's method}, 
J. Phys. A {\bf 46} (2013), 025201, 12 pp.

\bibitem{CMOQ2}
E. Celledoni, R.I. McLachlan, D.I. McLaren, B. Owren, G.R.W. Quispel.
{\em Integrability properties of Kahan's method}, 
J. Phys. A {\bf{47}} (2014), 365202, 20 pp.

\bibitem{CMOQ4}
E. Celledoni, R.I. McLachlan, D.I. McLaren, B. Owren, G.R.W. Quispel.
{\em Two classes of quadratic vector fields for which the Kahan map is integrable},
MI Lecture Note, Kyushu University {\bf 74} (2016), 60--62.

\bibitem{Dolgachev}
I.V. Dolgachev. 
{\em Rational surfaces with a pencil of elliptic curves},
Izv. Akad. Nauk SSSR Ser. Mat., {\bf 30}:5 (1966), 1073--1100.

\bibitem{Dui}
J.J. Duistermaat. 
{\em Discrete Integrable Systems. QRT Maps and Elliptic Surfaces},
Springer, 2010, xii+627 pp.

\bibitem{KMNOY}
K. Kajiwara, T. Masuda, M. Noumi, Y. Ohta, Y. Yamada.  
{\em Point configurations, Cremona transformations and the elliptic difference Painlev\'e equation},
S\'eminaires et Congr\`es {\bf 14} (2006), 169--198.

\bibitem{KNY}
K. Kajiwara, M. Noumi, Y. Yamada.
{\em Geometric aspects of Painlev\'e equations}, 
J. Phys. A: Math. Theor. {\bf 50} (2017) 073001 (164 pp).

\bibitem{KCMMOQ} 
P.H. van der Kamp, E. Celledoni, R.I. McLachlan, D.I. McLaren, B. Owren, G.R.W. Quispel.
{\em Three classes of quadratic vector fields for which the Kahan discretization is the root of a generalised Manin transformation},
J. Phys. A: Math. Theor. {\bf 52} (2019) 045204.

\bibitem{KMQ} 
P.H. van der Kamp, D.I. McLaren, G.R.W. Quispel.
{\em  Generalised Manin transformations and QRT maps}, {\tt  	arXiv:1806.05340 [nlin.SI].}

\bibitem{K}
W.~Kahan.
{\em Unconventional numerical methods for trajectory calculations},
Unpublished lecture notes, 1993.

\bibitem{HKY}
K. Kimura, H. Yahagi, R. Hirota, A. Ramani, B. Grammaticos, Y. Ohta.
{\em A new class of integrable discrete systems},
J. Phys. A: Math. Gen. {\bf 35} (2002) 9205--9212.

\bibitem{Manin}
Yu.I. Manin.
{\em The Tate height of points on an Abelian variety, its variants and applications},
Izv. Akad. Nauk SSSR Ser. Mat. {\bf 28}:6 (1964), 1363--1390.


\bibitem{PPS2}
M. Petrera, A. Pfadler, Yu.B. Suris.
{\em On integrability of Hirota-Kimura type discretizations},
Regular Chaotic Dyn. {\bf 16} (2011), No. 3-4, p. 245--289. 

\bibitem{PSS}
M.~Petrera, J.~Smirin, Yu.B.~Suris. 
{\em Geometry of the Kahan discretizations of planar quadratic Hamiltonian systems}, 
Proc. Royal Soc. A, {\bf 475} (2019), 20180761, 13 pp.

\bibitem{PS4}
M.~Petrera, Yu.B.~Suris. 
{\em Geometry of the Kahan discretizations of planar quadratic Hamiltonian systems. II. Systems with a linear Poisson tensor},
J. Comput. Dyn., {\bf 6} (2019), 401--408.

\bibitem{PSZ}
M. Petrera, Yu.B. Suris, R. Zander.
{\em How one can repair non-integrable Kahan discretizations},
J. Phys. A (2020, to appear), {\tt http://arxiv.org/abs/2003.12596}.


\bibitem{PZ}
M.~Petrera, R.~Zander,
{\em New classes of quadratic vector fields admitting integral-preserving Kahan-Hirota-Kimura discretizations},
J. Phys. A: Math. Theor. {\bf 50} (2017) 205203, 13 pp.

 
\bibitem{QRT}
G.R.W. Quispel, J.A.G. Roberts, C.J. Thompson. 
{\em  Integrable mappings and soliton equations II}, 
Physica D {\bf 34} (1989) 183--192.
 
\bibitem{Sakai}
H. Sakai. 
{\em Rational surfaces associated with affine root systems
and geometry of the Painlev\'e equations},
Commun. Math. Phys. {\bf 220} (2001), 165--229.

\bibitem{Tsuda}
T. Tsuda. {\em Integrable mappings via rational elliptic surfaces}, 
J. Phys. A: Math. Gen. {\bf 37} (2004), 2721--2730.

\bibitem{Ves}
A.P. Veselov.
{\em Integrable maps}, 
Russ. Math. Surv. {\bf 46} (1991) 3--45.

\bibitem{VGR}
C.M. Viallet, B. Grammaticos, A. Ramani.
{\em On the integrability of correspondences associated to integral curves},
Phys. Lett. A, {\bf 322} (2004) 186--193.

\bibitem{Z}
R. Zander.
{\em On the singularity structure of Kahan discretizations of a class of quadratic vector fields},
{\tt https://arxiv.org/abs/2003.01659}.


\end{thebibliography}
\end{document}